\documentclass{ieeeaccess}
\usepackage{cite}
\usepackage[hyphens]{url}
\usepackage{amsmath,amssymb,amsthm,booktabs,tabularx,array,graphicx,textcomp,microtype,enumitem,float}
\usepackage{etoolbox}
\definecolor{gray}{gray}{0.5}

\newcommand{\code}[1]{{\fontseries{m}\fontshape{n}\ttfamily #1}}
\newtheorem{proposition}{Proposition}
\graphicspath{{figures/}}
\makeatletter
\patchcmd{\@maketitle}{Digital Object Identifier\space\@doi}{}{}{}
\patchcmd{\IEEEbiographynophoto}{\vskip 4\baselineskip plus 1fil minus 0\baselineskip}{\vskip 1.5\baselineskip}{}{}
\newlength{\reviewabstractwidth}
\let\abstractwidth\reviewabstractwidth
\let\reviewheadings\ps@headings
\def\ps@headings{\reviewheadings
  \def\@oddfoot{{\footervolfont Manuscript for review\hfill\thepage}}
  \let\@evenfoot\@oddfoot}
\let\reviewtitlepage\ps@titlepage
\def\ps@titlepage{\reviewtitlepage
  \def\@oddfoot{{\footervolfont Manuscript for review\hfill\thepage}}
  \let\@evenfoot\@oddfoot}
\makeatother
\begin{document}
\pagestyle{headings}
\history{}
\doi{}
\vol{}\year{2026}
\title{GPU-Initiated Discrete Simulated Bifurcation: Low-Latency Requests and Streaming Dense Couplings}
\author{\uppercase{Yaocheng Chen}\authorrefmark{1}}
\address[1]{Instituto de Astronomia, Geof\'isica e Ci\^encias Atmosf\'ericas (IAG-USP), Universidade de S\~ao Paulo, Rua do Mat\~ao 1226, S\~ao Paulo, SP 05508-090, Brazil (e-mail: yaochengchen@usp.br)}
\markboth{Y. Chen: GPU-Initiated Discrete Simulated Bifurcation}{Y. Chen: GPU-Initiated Discrete Simulated Bifurcation}
\tfootnote{Yaocheng Chen is supported by a FAPESP Postdoctoral Fellowship (Grant No. 25/25586-5).}
\corresp{Corresponding author: Yaocheng Chen (e-mail: yaochengchen@usp.br).}
\begin{abstract}
GPU-based optimization faces two communication bottlenecks: coordinating frequent requests and delivering dense models that exceed device memory. We present a discrete simulated bifurcation (dSB) architecture that addresses both through NVIDIA DOCA GPUNetIO. For resident models, a persistent service receives field updates, executes each solve within one GPU thread block, and returns the result. Exact integer coupling sums, GPU work queues, and batched transmission keep the receive--solve--reply path on the device without a dedicated CPU data-path core. In comparisons with socket-based servers using the same solver, the largest latency gains occur under concurrent load. As the offered load increases from 400 to 800 thousand requests per second, median round-trip latency rises by only 6\%. At the highest tested load, median and 99th-percentile latencies are 189 and 218~$\mu$s, compared with 288 and 609~$\mu$s for the tuned persistent CPU proxy across repeated runs. For models larger than device memory, a streaming solver retains dynamical state on the GPU and reuses incoming coupling tiles across replicas. It evaluates ten-million-variable dense binary matrices at approximately 307~Gb/s, consuming a 12.5-TB logical matrix through a 64-MiB packet buffer. Ground-state recovery on planted instances and agreement with reference executions verify the computation. Together, the two modes scale dSB to concurrent requests and dense models beyond GPU memory.

\end{abstract}
\titlepgskip=-21pt
\maketitle

\section{Introduction}
Many combinatorial decisions admit Ising or quadratic unconstrained binary optimization (QUBO) formulations \cite{lucas,glover}. Simulated bifurcation (SB) solves these formulations through the dynamics of coupled nonlinear oscillators \cite{goto2019,goto2021}. SB is also used for repeated decisions driven by incoming observations. Examples include market-graph search for pairs trading, detection-to-track assignment in automotive perception, time--frequency resource allocation in wireless networks, and online beam selection \cite{pairstrading,tracking,wireless5g,rtbbo}. Reported systems include a 33-$\mu$s market-feed-to-order trading loop, vehicle-mountable tracking above 20 frames/s, and wireless resource allocation below 0.5~ms. The trading, tracking, and online beam-control systems use ballistic SB (bSB), while the wireless resource-allocation study reports an SB machine. These applications motivate efficient repeated execution; here we study the communication path for discrete SB (dSB).

In a dynamic decision process, instance $\ell$ has energy
\begin{equation}
E_{\ell}(s)=-\tfrac12 s^\top J_{\ell}s-h_{\ell}^\top s.
\label{eq:dynamic}
\end{equation}
The instance index $\ell$ denotes a new decision, separately from the integration steps within one solve. Some applications reuse $J_{\ell}$ and update only the field $h_{\ell}$; others update the quadratic coefficients as observations change. In online black-box optimization, feedback trains a quadratic surrogate, and an embedded Ising machine selects the next action from the updated model \cite{rtbbo}. The representation of $J_{\ell}$ determines the cost of constructing, storing, and delivering each instance, so model updates and communication consume part of the decision budget.

Graphics processing units (GPUs) offer a programmable environment for combining model construction, numerical optimization, and application processing. A conventional GPU service receives a request on the CPU, transfers its inputs, dispatches GPU work, and returns the result through the host network stack. For a short solve, this communication and dispatch can occupy a substantial fraction of the round trip. GPU-initiated networking provides a direct device-side path: GPUnet, GPUrdma, GPU-Ether, and related systems established the communication mechanisms \cite{gpunet,gpurdma,gpuether,janus}, and DOCA GPUNetIO exposes receive and transmit operations to CUDA kernels \cite{gpunetio}. An optimization service must additionally coordinate model reuse, persistent scheduling, and solver execution without host dispatch between integration steps. Higher network bandwidth also increases the demand for coordination, while gains in general-purpose single-thread performance have slowed \cite{hennessy2019}. Moving packet handling and solver dispatch onto the device reduces the host coordination required as arrival rates and queue counts grow. NVIDIA reports this coordination cost even for CPU-driven DPDK pipelines that already place packets directly in GPU memory \cite{gpunetioblog}.

These workloads expose two complementary scaling limits. At small problem sizes, scaling is primarily a request-rate problem: the model fits in GPU memory, but repeated host coordination can limit the number of concurrent decisions. At large problem sizes, scaling becomes a coefficient-delivery problem: the dense model requires $O(n^2)$ storage, while the $O(Bn)$ state of $B$ solver replicas can remain resident. The first regime calls for efficient receive--solve--reply coordination; the second calls for sustained delivery of coefficients to the solver.

Our approach is to bring incoming coefficients directly to GPU-resident dSB state and perform frequent scheduling on the device. A resident-model service accepts small field updates and executes each complete trajectory within one persistent worker block. We compare it with socket-based launch-per-request and CPU-proxied persistent servers using the same solver, protocol, and client, measuring latency under load and CPU data-path requirements. A streaming solver consumes nonresident coupling tiles while retaining positions, momenta, spins, and accumulators on the GPU. Each coupling evaluation requires a matrix pass. The request service schedules complete solves on the GPU; the streaming implementation schedules tiles on the GPU with host coordination at pass boundaries.

The contributions are:
\begin{itemize}[leftmargin=*]
\item A dSB solver that assigns each request to one block, retains dynamical state in registers and models in shared memory, and computes exact coupling sums with XOR/popcount or DP4A (Section~\ref{sec:solver}).
\item A resident-model protocol and GPU-resident receive--solve--reply pipeline with bounded work queues and batched transmission (Sections~\ref{sec:protocol} and~\ref{sec:server}).
\item An evaluation against two socket-based CPU-mediated implementations across 19 request configurations, examining dedicated CPU data-path cores, host-state dependence, and latency under load (Sections~\ref{sec:method} and~\ref{sec:results}).
\item A streaming dSB implementation that evaluates dense $n=10^7$ binary coupling matrices at 306--307~Gb/s for one or eight replicas, with planted-instance and reference-execution validation (Section~\ref{sec:stream}).
\end{itemize}

\section{Background and related work}\label{sec:related}
\subsection{Discrete simulated bifurcation}
SB minimizes the Ising energy
\begin{equation}
E(s) = -\tfrac12 \sum_{i\neq j} J_{ij} s_i s_j - \sum_i h_i s_i,\qquad s\in\{-1,+1\}^n,
\label{eq:energy}
\end{equation}
by integrating a classical system of positions $x_i$ and momenta $y_i$ through a bifurcation \cite{goto2019,goto2021}. In the discrete variant the coupling force depends only on $\mathrm{sign}(x)$. With detuning $\delta=1$, pump schedule $p_k=k/(K-1)$ for $K>1$ (and $p_0=0$ for $K=1$), coupling gain $\xi$ and time step $\Delta t$, one step is
\begin{align}
a_i &= \textstyle\sum_{j} J_{ij}\,\mathrm{sign}(x_j) + h_i, \label{eq:acc}\\
y_i &\leftarrow y_i + \bigl(-(\delta-p_k)x_i + \xi a_i\bigr)\Delta t, \notag\\
x_i &\leftarrow x_i + \Delta t\,\delta\, y_i. \label{eq:update}
\end{align}
followed by the inelastic wall: if $|x_i|>1$ then $x_i\leftarrow\mathrm{sign}(x_i)$ and $y_i\leftarrow 0$. The answer is $\mathrm{sign}(x)$ after $K$ steps; $B$ independent replicas are integrated from different initial states and the lowest-energy one is returned. We use $\mathrm{sign}(0)=-1$ and $\xi=0.5\sqrt{n-1}/\|J\|_F$ by default, as in \cite{ourrepo}.

\subsection{GPU and FPGA SB implementations}
Field-programmable gate array (FPGA) implementations of SB exploit specialized arithmetic pipelines and tightly integrated I/O to achieve low latency, with multi-chip designs extending the coupling computation \cite{goto2021,tatsumura2021,kashimata2024}. Trading systems demonstrate the value of this specialization in a complete decision loop \cite{tatsumura2023trading,yamasaki2020arb,tatsumura2021heart}. GPU implementations offer a complementary advantage: solver variants, coefficient formats, and surrounding numerical or learning stages can be changed within a common software environment. The companion GPU implementation \cite{ourrepo} provides whole-device and bit-packed execution paths. Here, an entire dSB request is confined to one block, allowing independent requests to execute concurrently while persistent device-side agents handle communication.

\subsection{GPU-centric networking}
GPUnet and GPUrdma support communication initiated by GPU threads \cite{gpunet,gpurdma}, and GPUDirect Async investigates device-driven communication on InfiniBand \cite{agostini2018}. GPU-Ether and Janus provide GPU-native packet I/O and virtualized packet-processing support on Ethernet \cite{gpuether,janus}. DOCA GPUNetIO supplies device APIs for receiving into GPU memory and posting sends \cite{gpunetio,gpudirect}; persistent kernels keep these device-side agents available between requests \cite{gupta2012}. NCCL's GPU-initiated networking also includes a GPUNetIO-backed transport \cite{hamidouche2025}, and a recent survey reviews GPU-centric communication methods \cite{unat2026}.

Kernel bypass and GPU initiation address different parts of the communication path. DPDK bypasses the operating-system network stack \cite{dpdk}; GPUdev can combine it with GPUDirect RDMA and persistent GPU processing while the CPU manages network queues and notifies the GPU \cite{gpudevblog}. GPUNetIO additionally permits CUDA kernels to control receive and transmit operations directly \cite{gpunetioblog}. Placing packets directly in GPU memory therefore does not by itself take the CPU out of orchestration. Host-driven designs also support high-throughput bulk ingest: Romein reaches high Ethernet rates for a GPU workload with a CPU-driven receive path \cite{romein2025}.

Our work integrates these communication mechanisms with coefficient delivery, resident solver state, and persistent scheduling for dSB. The evaluation covers request-rate scaling against socket-based servers using the same solver and model-size scaling through streamed coupling evaluation.

\section{Service protocol}\label{sec:protocol}
Every message is one raw Ethernet frame with a dedicated EtherType, a 64-byte header including the 14-byte Ethernet header, and a payload (Table~\ref{tab:proto}). Raw frames were chosen over UDP so that the network interface controller (NIC) can select traffic by EtherType alone and the GPU does not have to parse IP headers; the design carries over to UDP by matching on a port instead. A \code{SOLVE} frame for $n=128$ is 320 bytes and the largest frame in the request-service tests, a model fragment, is 1472 bytes, so the standard MTU suffices.

\begin{table}[!t]\centering\small
\caption{Message types. All fields after the Ethernet header are little-endian and naturally aligned.}
\label{tab:proto}
\begin{tabularx}{\columnwidth}{l>{\raggedright\arraybackslash}X}\toprule
Message & Payload \\ \midrule
\code{LOAD} & one fragment (\,$\le$1408~B) of a packed coupling model; header: model id, $n$, format, offset, fragment index and count, total size, generation, default $\xi$ \\
\code{LOAD\_ACK} & none; sent when all fragments of a generation have arrived \\
\code{SOLVE} & optional field $h$ ($n$ signed 16-bit integers); header: model id, $n$, $K$, $B$, seed, $\xi$, $\Delta t$ \\
\code{RESULT} & best spin configuration, $\lceil n/32\rceil$ words; header: status, $2E$, best replica, worker id, server-side timestamps (queue, solve, transmit) \\
\bottomrule
\end{tabularx}
\end{table}

A model is uploaded once and remains resident in a GPU model store of 64 slots of 512~KiB until it is replaced by a new generation. Three packed formats are supported (with $W=\lceil n/32\rceil$): \code{PM1}, one bit plane of $nW$ words for dense $\pm1$ couplings; \code{TERNARY}, a plane for $J_{ij}=+1$, a plane for $J_{ij}=-1$ and a row constant $|P_i|-|M_i|$; and \code{INT8}, dense int8 rows. The field travels with each request. This matches problems in which the quadratic part changes slowly and the linear part changes quickly. For mean--variance portfolio selection with a cardinality penalty \cite{markowitz1952},
\begin{equation}
\min_{x\in\{0,1\}^n} q\,x^\top\Sigma x-\mu^\top x+\lambda(\mathbf 1^\top x-k)^2,
\end{equation}
the substitution $x=(1+s)/2$ gives $J_{ij}=-Q_{ij}/2$ for $i\ne j$ and $J_{ii}=0$ and $h_i=-\tfrac12\bigl(\sum_j Q_{ij}+c_i\bigr)$ with $Q=q\Sigma+\lambda\mathbf 1\mathbf 1^\top$ and $c=-\mu-2\lambda k\mathbf 1$. With fixed covariance and penalty parameters, $J$ is reused, and each update of the expected returns $\mu$ gives a new $h$. The client quantizes $J$ to int8 and $h$ to int16 with a common scale. All solver and reference evaluations use the resulting quantized Ising objective, with cardinality enforced through the penalty term.

\section{Request-per-block solver}\label{sec:solver}
\subsection{Mapping}
One request is solved by one thread block of $T=256$ threads. With $n_p=\lceil n/32\rceil\cdot 32$, the $Bn_p$ pairs (replica $b$, row $r$) are distributed cyclically, pair $p=bn_p+r$ to thread $p \bmod T$, and each thread keeps the $x$ and $y$ of its at most $P=16$ pairs in registers for the whole trajectory; hence $Bn_p\le TP=4096$. Because $n_p$ and $T$ are multiples of 32, the 32 lanes of a warp always hold 32 consecutive rows of the same replica. After the update the new signs of a warp are therefore exactly one 32-bit word of that replica's sign vector: \code{\_\_ballot\_sync} forms the word and lane~0 stores it, without atomics or shuffles. Sign vectors are double buffered in shared memory, so a step needs a single \code{\_\_syncthreads()}: all warps read buffer $k \bmod 2$ and write buffer $(k+1)\bmod 2$.

Shared memory holds the coupling rows, padded to a stride $W_p=4\cdot\text{odd}\ge W$ words so that the 16-byte loads of eight consecutive lanes fall in distinct banks, the two sign buffers of $B$ rows each, the field and the per-replica energies. The footprint is
\begin{equation}
S = 4n_pW_p\,c + 8BW_p + 4n_p + 8B + S_0,
\label{eq:smem}
\end{equation}
with $c=1$ for \code{PM1}, $c=2$ (plus $4n_p$ for the row constants) for \code{TERNARY}, and $S_0=1024$~B for the request header and the reply staging area; for \code{INT8} the coupling term is $n_pL_s$ with a byte stride $L_s=16\cdot\text{odd}\ge n_p$ and the sign buffers are $2BL_s$ bytes. The server fixes one $S$ for all workers, which sets how many blocks are resident per streaming multiprocessor (SM) and so the number of requests in service at once; with $S=48$~KiB two blocks fit on each of the 188 SMs of the server GPU, 376 in all.

\subsection{Exact coupling sums}
\begin{proposition}\label{prop:exact}
Let $J$ be integer with zero diagonal, $s\in\{\pm1\}^n$ and $v_j=\mathbf 1[s_j=+1]$ packed into words.
(i) If $J_{ij}=2u_{ij}-1$ for $j\neq i$ (bit $u_{ij}=1$ for $J_{ij}=+1$) with $u_{ii}=0$ and zero padding, then $\sum_{j} J_{ij}s_j = n-2\,\mathrm{popc}(u_i\oplus v)+s_i$.
(ii) If $J_{ij}\in\{-1,0,1\}$ with $P_{ij}=\mathbf 1[J_{ij}=1]$ and $M_{ij}=\mathbf 1[J_{ij}=-1]$, then
\[
\begin{aligned}
\sum_j J_{ij}s_j
&=2\bigl(\mathrm{popc}(P_i\wedge v)-\mathrm{popc}(M_i\wedge v)\bigr)\\
&\quad-\bigl(|P_i|-|M_i|\bigr).
\end{aligned}
\]
(iii) If $|J_{ij}|\le127$, $|h_i|<2^{15}$ and $n\le 2048$, the int32 DP4A sum $a_i$ of \eqref{eq:acc} is exact and $|a_i|<2^{24}$, so $a_i$ is exact in FP32.
Consequently, if \eqref{eq:update} is evaluated with the same IEEE round-to-nearest operations in the same order, the GPU and a CPU reference produce identical trajectories, spins and energies.
\end{proposition}
\begin{proof}
(i) A column $j<n$, $j\neq i$ contributes $+1$ when $u_{ij}=v_j$ and $-1$ otherwise; the $32W-n$ padding columns have $u=v=0$ and contribute $+1$ each; the diagonal, stored as $u_{ii}=0$, contributes $+1$ if $v_i=0$ and $-1$ otherwise, that is $-s_i$. Summing over all $32W$ columns gives $32W-2\,\mathrm{popc}(u_i\oplus v)$, and removing the padding and diagonal terms gives the identity. (ii) $\sum_j P_{ij}s_j=2\,\mathrm{popc}(P_i\wedge v)-|P_i|$ and likewise for $M$. (iii) $|a_i|\le 127\cdot2047+2^{15}<2^{24}$. The final statement follows because $a_i$ is the same integer on both sides and every remaining operation is a correctly rounded IEEE operation on identical operands.
\end{proof}
The GPU update uses explicitly rounded multiplication, addition, and subtraction (\code{\_\_fmul\_rn}, \code{\_\_fadd\_rn}, and \code{\_\_fsub\_rn}). These operations prevent fused multiply--add contraction. The host reference is compiled with \code{-ffp-contract=off}. Initial states are drawn from a counter-based hash of (seed, replica, row), so the reference can regenerate them. The self-test and the client use this property for the energy and reference-output checks specified in Section~\ref{sec:method}. The equivalence statement uses identical initial states, parameters, tie breaking, rounding, and treatment of subnormal values.

\subsection{Cost of one request}
A \code{PM1} step issues $Bn\cdot 4\lceil W/4\rceil$ population counts per block (twice as many AND/POPC pairs for \code{TERNARY}), and an \code{INT8} step $Bn\,n_p/4$ DP4A instructions. With an instruction throughput $r$ per SM per cycle and clock frequency $f$, the solve time is approximately
\begin{equation}
t_{\mathrm{solve}}\approx t_{\mathrm{load}} + K\Bigl(\frac{\mathrm{ops}}{r\,f\,/\,\beta} + t_{\mathrm{bar}}\Bigr),
\label{eq:tsolve}
\end{equation}
where $\mathrm{ops}$ is the instruction count per step and $\beta$ is the number of blocks sharing the SM, $t_{\mathrm{bar}}$ the cost of a block barrier and $t_{\mathrm{load}}$ the one-time copy of the model rows from L2 into shared memory. Table~\ref{tab:single} gives the measured values: the time per step is 0.5--0.6~$\mu$s for the smallest blocks, where the barrier and the instruction latency rather than the throughput of the SM set the pace, and grows to 4.7~$\mu$s for $n=128$, $B=32$, \code{INT8}. A 100-step request on 128 variables with up to eight replicas therefore costs 60--150~$\mu$s, so the 10--19~$\mu$s of host-side work that Section~\ref{sec:server} removes is a noticeable fraction of a short request.

\section{Server architectures}\label{sec:server}
\begin{figure*}[tp]\centering
\includegraphics[width=\textwidth]{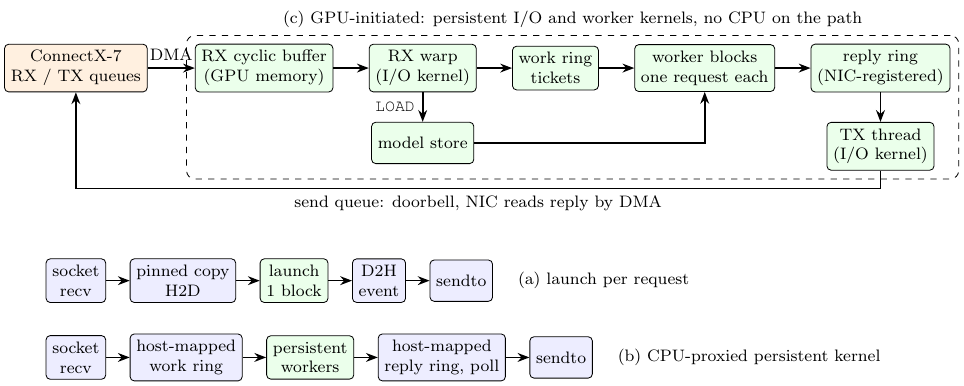}
\caption{The three servers. All share the wire format, the client and the per-request solver; they differ only in how a frame reaches a worker block and how the reply leaves. Blue: CPU; green: GPU; orange: NIC.}
\label{fig:arch}
\end{figure*}

Figure~\ref{fig:arch} shows the three servers compared in this paper. The two persistent designs share a worker loop and the ring organization shown in Fig.~\ref{fig:arch}(c): a work ring of $C=4096$ items, each pointing at a frame and carrying a sequence number and a receive timestamp, and a reply ring of 4096 slots of 512~B. A worker takes a ticket $t$ with one atomic increment, waits for item $t\bmod C$ to carry sequence number $t+1$, copies the request header and field into shared memory, solves, allocates a reply slot, writes the reply and publishes a ready word that contains the slot's sequence number and the reply length. What differs is who fills the work ring and who drains the reply ring.

\subsection{Launch per request}
The conventional design uses a nonblocking \code{AF\_PACKET} socket to receive each frame, copies it into pinned memory, and on one of $S_t=4$ CUDA streams issues a host-to-device copy, a one-block launch of the solver, a device-to-host copy of the reply and an event; completed events are polled and their replies sent. Up to $S_t$ requests are in flight. Models are reassembled on the CPU and copied to the GPU model store.

\subsection{CPU-proxied persistent kernel}
The second design removes launches and explicit copies but keeps the CPU on the data path. The CPU writes each frame into a pool of host-mapped pinned buffers and publishes a work item in a host-mapped ring; 376 persistent worker blocks, running the same worker loop as the GPU-initiated server, read the request over PCIe and write the reply into a host-mapped reply ring, which the CPU polls and sends. The same mapped-buffer organization would also work with kernel-bypass networking; here it uses the kernel socket path.

\subsection{GPU-initiated server}
The third design keeps the request data path on the GPU after startup. The host creates a DOCA Flow root pipe that steers the service EtherType to one receive queue, a cyclic GPUNetIO receive queue whose 16,384-slot buffer (64~MiB) lives in GPU memory, and a send queue of depth 4096 whose completions are also handled on the GPU; the reply ring is allocated in GPU memory and registered with the NIC through a DMA-BUF. It then launches two kernels on two streams, an I/O kernel of two 32-thread blocks first, so that it is resident before the workers fill the GPU, and the worker kernel with $N_w=372$ blocks second (four of the 376 resident slots are left free so that the device is never completely booked by persistent blocks). The kernels communicate only through the rings and need no cooperative launch.
\begin{itemize}[leftmargin=*]
\item \emph{The receive warp} (block 0 of the I/O kernel) calls the warp-scope receive of GPUNetIO in a loop, taking up to 64 frames per call with a timeout of 500~ns: the call returns when 64 frames have arrived or when the timeout has expired. The timeout cannot be zero, because on this platform a zero timeout blocks until a frame arrives and the warp would never see the shutdown flag on an idle link; our first runs used 5~$\mu$s. It reads the 64-byte header of each frame with L2-coherent loads, validates it, and for a \code{SOLVE} frame copies the frame out of the NIC's cyclic buffer into a frame buffer owned by the work slot (a 320-byte warp copy for $n=128$), stamps the receive time from the GPU global timer and publishes the item. For a \code{LOAD} fragment the warp copies the payload into the model store and, when all fragments of a generation have arrived, publishes the model and queues an acknowledgement.
\item \emph{Worker blocks} run the common worker loop; the worker that solves a request stamps its claim time and the end of the solve into the reply.
\item \emph{The transmit thread} (thread 0 of block 1 of the I/O kernel) scans the reply ring in order and stamps the time each reply has waited since the end of its solve. It then writes one send descriptor per ready reply, up to 32, requests a completion only for the last, rings the doorbell once, and waits for that completion before it releases the slots. The NVIDIA \code{gpunetio\_simple\_send} sample offers this sequence and a simpler one in which every reply is posted with its own send call and completion notification; our first version used the simpler one, and Section~\ref{sec:load} reports both, because the choice sets the server's capacity.
\end{itemize}
Three rules make the rings safe without locks. First, the NIC recycles its receive buffer at line rate regardless of how far the workers are behind, so the receiver copies each request before publishing it. Second, a worker acknowledges its ticket in the slot after its reply is published, and the receiver reuses a slot only when the previous ticket of that slot has been acknowledged; otherwise the request is dropped and counted, which bounds the backlog at $C=4096$ requests. Third, reply slots are allocated after the solve rather than at arrival, so a long request does not delay the replies of shorter requests that finish after it. Producers write an item and then its sequence number behind a memory fence, and consumers read requests and models with L2-coherent loads (\code{\_\_ldcg}), because the data were written by another SM. Model fragments are tracked by index in a per-generation bitmap, so a retransmitted fragment is ignored and a retransmission after completion re-sends the acknowledgement.

Host monitoring and GPU occupancy affect latency. The host must be able to observe the server (counters, shutdown) without touching the data path: the request-service counters live in managed memory and are read once per second, and the shutdown flag lives in host-mapped memory. In a first version every spinning block read that flag on every iteration; 372 blocks issuing PCIe reads in a tight loop starved the NIC's own DMA and doorbells and raised the round-trip time to milliseconds. Each block now polls the flag no more than once every 0.5~ms, reducing PCIe traffic. Four block slots are left free for driver work.

\subsection{Where the time goes}
The round-trip time seen by a client decomposes as
\begin{equation}
T = t_{\mathrm{wire}} + t_{\mathrm{in}} + t_{q} + t_{\mathrm{solve}} + t_{\mathrm{out}}.
\label{eq:decomp}
\end{equation}
For the launch path $t_{\mathrm{in}}$ contains the NIC's delivery to the operating-system kernel, the socket receive, a pinned copy, the host-to-device copy and the launch, and $t_{\mathrm{out}}$ the device-to-host copy, the completion poll and the socket send. For the proxy path they contain the operating-system receive path, the socket calls and the PCIe round trips of the mapped rings. For the GPU-initiated path they reduce to a direct memory access (DMA) write, the receive poll, a ring hand-off through L2 and a doorbell. Every server stamps its residence time into the reply. The GPU-initiated server stamps $t_q$ (receive to claim), $t_{\mathrm{solve}}$ and the transmit wait with the GPU global timer; the two baselines stamp the interval from the return of \code{recv()} to the call of \code{sendto()} with the CPU clock and report it as $t_{\mathrm{solve}}$ plus a host term. The remainder, $T$ minus recorded server residence, includes wire transit and the client's network processing. Its server-side contribution follows the timestamp boundaries: NIC-to-GPU DMA and the receive-call wait for GPUNetIO, or kernel processing before \code{recv()} returns and after \code{sendto()} is called for the baselines. Section~\ref{sec:latency} uses this decomposition to interpret the end-to-end measurements.

\section{Experimental methodology}\label{sec:method}
\subsection{Platform}
The server is an x86-64 workstation with eight NVIDIA RTX PRO 6000 Blackwell Workstation Edition GPUs, of which one is used, and a ConnectX-7 NIC; the client is a dual-socket GH200 system, whose GPUs are idle in the request-service campaign and one of which is the generator of the streaming mode (Section~\ref{sec:stream}); the two are connected back to back by 400~GbE (Table~\ref{tab:config}). The server GPU is attached with a 256~MiB BAR1 window (resizable BAR disabled), which is why the receive ring is 64~MiB. The client and the baseline servers pin their network threads to dedicated cores. The CPU-mediated servers are measured in two host states. In the \emph{default} state interrupt moderation, processor idle states and the interrupt affinity of the NICs are left at the distribution's defaults. In the \emph{tuned} state a script of the repository disables every processor idle state deeper than C0 on all cores, selects the performance governor, turns interrupt moderation, GRO and LRO off and pins the NIC's interrupts next to the polling core, on the server and on the client alike. The GPU-initiated server is reported with the tuned client. In a preliminary comparison using per-reply sends, its reference-request round trip was 168.5~$\mu$s in the default state and 168.6~$\mu$s in the tuned state, so we report it in one host state.

\begin{table}[!t]\centering\small
\caption{Configuration of the measurement campaign.}
\label{tab:config}
\begin{tabularx}{\columnwidth}{l>{\raggedright\arraybackslash}X}\toprule
Item & Value \\ \midrule
Server GPU & RTX PRO 6000 Blackwell Workstation Edition (\code{sm\_120}, 188 SMs), driver 595.58.03 (reports CUDA 13.2 support) \\
Server NIC & ConnectX-7, firmware 28.48.1000, MLNX\_OFED 26.04-0.8.6, DOCA 3.4.0112 \\
Server OS & Linux 5.14.0-570.62.1.el9 (RHEL 9.6); baseline network thread pinned to one core \\
Client & GH200 (2$\times$ 72 Neoverse-V2 cores), Linux 7.0.0-1019-nvidia-64k, ConnectX-7 firmware 28.43.8002; TX and RX threads on two dedicated cores, \code{AF\_PACKET} \\
Link & 400~GbE direct attach, MTU 9000 (frames $\le$1472~B in the request service, 4096~B frames with 4032~B of payload in the streaming mode) \\
Streaming mode & generator on GH200 GPU~1: 4 send queues (1--8 in one sweep), 96 synthesis blocks by default; receiver: one receive queue of 16,384 slots of 4~KiB (64~MiB), one receive warp, 372 worker blocks of 512 threads; $B\in\{1,8\}$ \\
Worker blocks & 256 threads, 48~KiB shared memory, 2 per SM: 372 (GPUNetIO), 376 (proxy) \\
Launch server & 4 CUDA streams \\
Host states & default; tuned (no processor idle states, performance governor, no interrupt moderation, interrupts pinned) \\
Network runs & 100,000 requests after 1,000 warm-up; 291 runs \\
GPU-only runs & 500 requests per single-launch point; 200,000 (GPU rings) and 100,000 (CPU rings) per pipeline point \\
\bottomrule
\end{tabularx}
\end{table}

The server uses PCIe~5.0, and its topology places the GPU in NUMA node~1. GPU-to-NIC paths are reported as \code{SYS} for \code{mlx5\_0}--\code{mlx5\_3} and \code{NODE} for \code{mlx5\_4}--\code{mlx5\_7}: the former traverse the inter-socket interconnect, and the latter traverse host bridges within the same NUMA node. All reported paths therefore include host-side interconnects. NVIDIA's GPUDirect RDMA guidance identifies a path confined to PCIe switches as the preferred topology \cite{gpudirect}. GPU/NIC placement and PCIe switching are consequently important configuration choices for scaling the streaming receiver.

\subsection{Workloads}
Three instance families cover the three formats: dense $\pm1$ Sherrington--Kirkpatrick couplings (\code{PM1}), sparse ternary couplings at 5\% density (\code{TERNARY}), and the int8 portfolio model of Section~\ref{sec:protocol} with a four-factor covariance, a cardinality target of $n/10$ and a pool of 256 field vectors generated from AR(1) expected returns and reused cyclically (\code{INT8}). The GPU-only sweep covers $n\in\{64,128,192,256,512\}$, $B\in\{1,8,32\}$ and $K\in\{50,100,200,400\}$ where Eq.~\eqref{eq:smem} admits them, 120 configurations in all. The network sweep uses a scaling grid of 19 configurations: \code{INT8} with $n\in\{64,128,192\}$ and $K\in\{50,100,200,400\}$ at $B=8$; \code{PM1} with $n\in\{64,128,256,512\}$ at $B=8$, $K=100$; \code{TERNARY} at $n=256$; and \code{INT8} $n=128$ with $B\in\{1,32\}$. The reference configuration for the load and distribution experiments is \code{INT8}, $n=128$, $B=8$, $K=100$. Every request has its own seed. The supplied request sweep uses the client defaults $\Delta t=1.0$ and base seed 1, with model seed $1000+m$ for model index $m$ and request seed $7919+i$ for request index $i$. The model's default coupling gain is used. With these settings the archived sweep is deterministic and can be reproduced.

\subsection{Measurements}
\emph{GPU only.} \code{dsb\_bench} measures, without a NIC, (i) the in-kernel solve time of one request from the GPU global timer, (ii) the launch-to-completion time from CUDA events and (iii) the host round trip H2D + launch + D2H from the CPU clock; it also saturates the persistent pipeline with a GPU generator warp and a GPU consumer warp that use the same rings as the GPU-initiated server, and the proxied pipeline with a CPU producer and consumer that, like the proxy server, drain all ready replies in each iteration and keep up to 4096 requests outstanding.
\emph{Network.} \code{dsb\_client} measures the interval from the client timestamp immediately before request-frame construction and \code{sendto} to the timestamp after \code{recv} in closed loop with 1, 4, 16 and 64 outstanding requests and in open loop with Poisson arrivals at 10 to 800~kreq/s. The client sends from one thread; at 700 and 800~kreq/s offered with the tuned client its own send timestamps show that it emitted 596--681~kreq/s, which is therefore the highest load of the campaign. The first 1000 requests of every run are discarded. We report the 50th, 99th and 99.9th percentiles; the servers' own timestamps split each round trip into server residence and the remainder of Section~\ref{sec:server}. By construction, the launch and proxy servers each occupy one busy-polling CPU core on the data path and the GPU-initiated server none. Checks are sampled by request index: energy is recomputed from returned spins for up to the first 3000 requests of each run (including warm-up), and a CPU integration checks the final result for the first 20. 

\subsection{Baseline scope}
All three servers share the solver, workloads, protocol, and client. Both baselines use the kernel socket path, with tuning of processor idle states, interrupts, and offloads. The comparison measures the combined effects of the network stack and CPU--GPU coordination (Section~\ref{sec:scope}).

\subsection{Reporting policy}
The archive contains 291 network runs covering 246 distinct settings. For each server, workload, arrival rate, and window, the analysis uses the latest run; the unloaded comparison uses 19 settings per server configuration. For GPUNetIO and the tuned proxy at 400--800~kreq/s offered, we report medians across repeated runs with the current code. Table~\ref{tab:repeats} gives the run counts and ranges. Latency percentiles are computed within each run before aggregation.

The configured Poisson rate is termed ``offered load,'' and the measured completion rate is termed ``achieved rate.'' At the two highest target rates, the client injects 596--681~kreq/s. Latency percentiles are computed over completed requests, and loss is reported separately.

\section{Results}\label{sec:results}
The repository includes the benchmark scripts, analysis scripts, and per-request records used to produce these results.

\subsection{Correctness}
The self-test compares the GPU with the CPU reference on eleven configurations across the three formats, through both the single-launch and persistent paths. All eleven configurations and 44 replies through the persistent ring passed before the measurement campaign. Across the 291 network runs, the client recomputed 850,874 reply energies and compared 5,820 final results with CPU-reference integrations, with no mismatch on any server.

\subsection{Cost of one request on the GPU}\label{sec:gpucost}
Table~\ref{tab:single} and Fig.~\ref{fig:overhead} separate the in-kernel solve time from the cost of the launch path without the network. Over the 120 configurations the solve takes 31~$\mu$s to 1.9~ms and grows linearly with $K$. Launching the block costs a further 3.7--7.2~$\mu$s (median 6.0) and the full host round trip of host-to-device (H2D) copy, launch, and device-to-host (D2H) copy costs 10.0--19.3~$\mu$s above the solve (median 13.0), across the tested requests. Relative to the solve, this overhead adds 60\% to the smallest request (31~$\mu$s), 8\% to the reference portfolio request (152~$\mu$s), and 1\% to a 1-ms request. These measurements use an idle host with a polling thread and exclude network I/O; Section~\ref{sec:latency} reports the end-to-end latency.

\begin{figure}[!t]\centering
\includegraphics[width=\columnwidth]{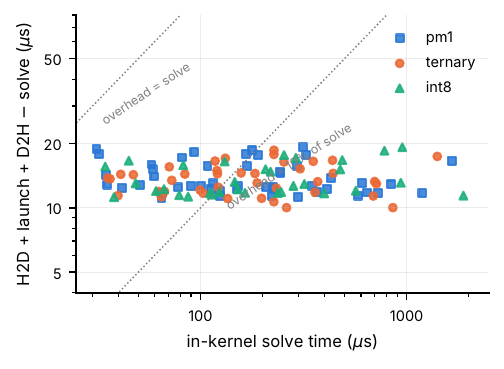}
\caption{Launch-path overhead (H2D + launch + D2H minus the in-kernel solve) against the solve time, one point per configuration, no network. The overhead is 10--19~$\mu$s for every configuration.}
\label{fig:overhead}
\end{figure}

\begin{table*}[tp]\centering\footnotesize
\caption{One request on the server GPU, no network (median over 500 requests). All times are in $\mu$s. The last two columns are the launch-to-completion time and the H2D + launch + D2H round trip, each minus the in-kernel solve, as the median over $K$ (range in parentheses).}
\label{tab:single}
\begin{tabular}{lrrrrrrrr}\toprule
 & & & \multicolumn{4}{c}{in-kernel solve, p50 ($\mu$s)} & launch & H2D+launch+D2H \\
\cmidrule(lr){4-7}
format & $n$ & $B$ & $K{=}50$ & 100 & 200 & 400 & $-$solve & $-$solve ($\mu$s) \\ \midrule
pm1 & 64 & 1 & 31.3 & 58.6 & 112.8 & 221.1 & 6.0 & 13.7 (11.4--18.8) \\
pm1 & 64 & 8 & 41.7 & 77.7 & 149.6 & 293.2 & 5.8 & 12.4 (12.2--15.8) \\
pm1 & 64 & 32 & 89.9 & 168.7 & 326.5 & 641.5 & 6.1 & 14.2 (11.8--17.6) \\
pm1 & 128 & 1 & 32.2 & 59.8 & 114.9 & 225.1 & 6.1 & 13.5 (12.5--17.9) \\
pm1 & 128 & 8 & 57.9 & 108.1 & 208.7 & 409.7 & 5.7 & 14.1 (12.2--16.0) \\
pm1 & 128 & 32 & 166.6 & 313.3 & 606.1 & 1192.5 & 6.3 & 15.5 (11.8--19.3) \\
pm1 & 192 & 1 & 34.9 & 64.7 & 124.3 & 243.4 & 6.0 & 13.1 (11.1--14.7) \\
pm1 & 192 & 8 & 81.2 & 152.9 & 296.2 & 582.6 & 6.1 & 12.0 (11.3--17.2) \\
pm1 & 256 & 1 & 35.4 & 65.2 & 124.9 & 244.4 & 6.5 & 12.3 (11.4--14.5) \\
pm1 & 256 & 8 & 100.8 & 190.3 & 369.1 & 726.7 & 6.6 & 12.3 (11.7--17.7) \\
pm1 & 512 & 1 & 50.9 & 93.4 & 178.1 & 347.8 & 5.6 & 15.5 (12.7--18.6) \\
pm1 & 512 & 8 & 224.4 & 428.7 & 837.2 & 1654.7 & 6.0 & 13.3 (12.2--16.5) \\
ternary & 64 & 1 & 35.5 & 62.8 & 117.7 & 227.2 & 5.7 & 15.2 (11.9--18.6) \\
ternary & 64 & 8 & 47.3 & 84.0 & 157.5 & 304.4 & 5.7 & 14.5 (14.3--15.2) \\
ternary & 64 & 32 & 103.0 & 188.3 & 358.9 & 699.5 & 5.9 & 12.5 (11.7--13.3) \\
ternary & 128 & 1 & 36.4 & 64.6 & 120.9 & 233.5 & 6.0 & 13.0 (11.3--14.4) \\
ternary & 128 & 8 & 67.0 & 120.3 & 227.1 & 440.3 & 5.7 & 14.7 (11.9--17.8) \\
ternary & 128 & 32 & 197.9 & 370.4 & 713.8 & 1408.1 & 6.4 & 13.1 (11.1--17.4) \\
ternary & 192 & 1 & 39.8 & 70.6 & 132.0 & 254.1 & 6.0 & 16.0 (11.4--17.1) \\
ternary & 192 & 8 & 99.6 & 183.8 & 352.2 & 689.2 & 6.0 & 13.3 (11.4--16.5) \\
ternary & 256 & 1 & 41.1 & 72.6 & 135.5 & 261.5 & 6.5 & 12.2 (10.0--14.3) \\
ternary & 256 & 8 & 121.9 & 227.0 & 437.1 & 857.4 & 6.0 & 11.6 (10.0--16.7) \\
int8 & 64 & 1 & 34.5 & 60.9 & 113.5 & 218.8 & 6.2 & 13.4 (11.5--15.6) \\
int8 & 64 & 8 & 48.5 & 87.1 & 164.4 & 318.9 & 6.0 & 12.3 (11.3--13.0) \\
int8 & 64 & 32 & 111.3 & 206.8 & 398.0 & 780.3 & 6.2 & 13.4 (11.7--18.5) \\
int8 & 128 & 1 & 38.1 & 66.5 & 123.4 & 237.2 & 6.2 & 11.8 (11.2--12.2) \\
int8 & 128 & 8 & 82.5 & 151.7 & 290.1 & 566.9 & 5.8 & 14.0 (12.0--17.1) \\
int8 & 128 & 32 & 254.3 & 487.1 & 953.2 & 1885.1 & 6.2 & 17.2 (11.4--19.2) \\
int8 & 192 & 1 & 44.9 & 78.8 & 146.7 & 282.3 & 6.0 & 12.9 (11.5--16.6) \\
int8 & 192 & 8 & 131.1 & 246.2 & 476.4 & 936.9 & 5.8 & 14.1 (11.9--16.5) \\
\bottomrule\end{tabular}

\end{table*}

\subsection{Pipeline capacity without the network}
Table~\ref{tab:pipeline} compares the saturated throughput of the GPU-resident rings, fed by a GPU generator warp, with that of the host-mapped rings fed by a CPU thread. For every configuration whose solve is short enough, the GPU-resident pipeline sustains 1.9--2.0~Mreq/s, the rate at which one generator lane can publish items; 62 of the 120 configurations reach this ceiling. For longer solves the workers become the limit: the throughput reaches 73--96\% of $N_w/t_{\mathrm{solve}}$, the ideal rate for the GPU-only benchmark's $N_w=375$ workers, the shortfall being the slowdown of two blocks sharing an SM. The CPU-proxied pipeline sustains 165--1025~kreq/s. For solves below 150~$\mu$s it spends 1.0--1.6~$\mu$s of CPU time per request, and the producer and consumer loops are the limit; for the reference configuration the ratio is 1983 to 658~kreq/s, 3.0$\times$. For the longest solves both pipelines approach the worker bound. The GPU-only benchmark reserves one of the 376 resident block slots for its generator/consumer control work, leaving 375 solver workers; the GPUNetIO server reserves four slots for communication and uses 372.

\begin{table*}[tp]\centering\footnotesize
\caption{Saturated persistent pipelines without the network, $B=8$, default host state. The GPU-only benchmark uses $N_w=375$ workers (the networked GPUNetIO server uses 372). $N_w/t_{\mathrm{solve}}$ is the bound at the single-request solve time. The CPU proxy drains all ready replies in each iteration and keeps up to 4096 requests outstanding.}
\label{tab:pipeline}
\begin{tabular}{lrrrrrrrr}\toprule
 & & & & \multicolumn{3}{c}{GPU-resident rings (GPU producer/consumer)} & \multicolumn{2}{c}{host-mapped rings (CPU proxy)} \\
\cmidrule(lr){5-7}\cmidrule(lr){8-9}
format & $n$ & $B$ & $K$ & kreq/s & \shortstack{$N_w/t_{\mathrm{solve}}$\\(kreq/s)} & p99 ($\mu$s) & kreq/s & p99 ($\mu$s) \\ \midrule
pm1 & 64 & 8 & 100 & 1990 & 4826 & 543.0 & 769 & 5489.9 \\
pm1 & 64 & 8 & 400 & 1224 & 1279 & 3347.7 & 595 & 7378.3 \\
pm1 & 128 & 8 & 100 & 1981 & 3468 & 577.3 & 721 & 6127.4 \\
pm1 & 128 & 8 & 400 & 868 & 915 & 4735.8 & 558 & 8124.5 \\
pm1 & 192 & 8 & 100 & 1974 & 2453 & 621.0 & 658 & 6710.1 \\
pm1 & 192 & 8 & 400 & 599 & 644 & 6889.4 & 472 & 9394.0 \\
pm1 & 256 & 8 & 100 & 1841 & 1971 & 2230.7 & 623 & 7078.5 \\
pm1 & 256 & 8 & 400 & 483 & 516 & 8547.3 & 418 & 10966.7 \\
pm1 & 512 & 8 & 100 & 823 & 875 & 5014.1 & 459 & 9553.2 \\
pm1 & 512 & 8 & 400 & 211 & 227 & 19706.2 & 211 & 19639.4 \\
ternary & 64 & 8 & 100 & 2030 & 4464 & 585.9 & 1025 & 4283.5 \\
ternary & 64 & 8 & 400 & 1154 & 1232 & 3557.6 & 560 & 7814.1 \\
ternary & 128 & 8 & 100 & 1974 & 3117 & 593.8 & 674 & 6548.8 \\
ternary & 128 & 8 & 400 & 808 & 852 & 5087.7 & 546 & 8717.9 \\
ternary & 192 & 8 & 100 & 1891 & 2040 & 2165.2 & 624 & 7059.2 \\
ternary & 192 & 8 & 400 & 497 & 544 & 8342.2 & 418 & 11325.1 \\
ternary & 256 & 8 & 100 & 1502 & 1652 & 2754.8 & 578 & 7621.1 \\
ternary & 256 & 8 & 400 & 394 & 437 & 10471.1 & 353 & 13794.6 \\
int8 & 64 & 8 & 100 & 2000 & 4303 & 569.1 & 997 & 4309.6 \\
int8 & 64 & 8 & 400 & 1007 & 1176 & 4082.3 & 583 & 7512.1 \\
int8 & 128 & 8 & 100 & 1983 & 2471 & 667.9 & 658 & 6732.1 \\
int8 & 128 & 8 & 400 & 536 & 661 & 7696.7 & 446 & 10090.5 \\
int8 & 192 & 8 & 100 & 1139 & 1523 & 3685.3 & 517 & 8517.0 \\
int8 & 192 & 8 & 400 & 291 & 400 & 14531.4 & 289 & 14467.6 \\
\bottomrule\end{tabular}

\end{table*}

\subsection{End-to-end latency}\label{sec:latency}
Table~\ref{tab:net} reports unloaded round-trip time (RTT) for the five server configurations. For the reference request, the GPUNetIO median is 166.4~$\mu$s, compared with 175.6~$\mu$s for tuned proxy and 178.5~$\mu$s for tuned launch. The improvement over the faster tuned baseline is approximately 5\%. The recorded solve times are nearly identical (149.9--152.9~$\mu$s), so the difference comes from the request path around the solver.

Table~\ref{tab:breakdown} reports component percentiles; Fig.~\ref{fig:breakdown} uses component means for an additive decomposition. On GPUNetIO, the median work-ring and transmit waits are 2.0 and 0.8~$\mu$s. The baselines' median host terms are 11.3--16.1~$\mu$s, which is roughly what the GPU-initiated path saves in server residence.

\begin{table*}[tp]\centering\footnotesize
\caption{Unloaded round-trip time in $\mu$s (closed loop, one outstanding request, 100,000 requests per run), measured on the client. ``Default'' and ``tuned'' are the two host states of Section~\ref{sec:method}; the GPU-initiated server uses the batched send path. Bold: the lowest median of the row.}
\label{tab:net}
\begin{tabular}{lrrrrrrrrrrrrr}\toprule
 & & & & \multicolumn{2}{c}{launch, default} & \multicolumn{2}{c}{proxy, default} & \multicolumn{2}{c}{launch, tuned} & \multicolumn{2}{c}{proxy, tuned} & \multicolumn{2}{c}{GPUNetIO} \\
\cmidrule(lr){5-6}\cmidrule(lr){7-8}\cmidrule(lr){9-10}\cmidrule(lr){11-12}\cmidrule(lr){13-14}
format & $n$ & $B$ & $K$ & p50 & p99 & p50 & p99 & p50 & p99 & p50 & p99 & p50 & p99 \\ \midrule
int8 & 64 & 8 & 50 & 71.7 & 78.3 & 73.9 & 80.8 & 67.7 & 70.6 & 70.9 & 75.6 & \textbf{62.1} & 65.1 \\
int8 & 64 & 8 & 100 & 110.4 & 121.3 & 113.5 & 208.2 & 106.6 & 108.9 & 110.0 & 114.5 & \textbf{101.6} & 103.5 \\
int8 & 64 & 8 & 200 & 283.1 & 295.0 & 278.4 & 292.8 & 182.9 & 192.1 & 188.5 & 192.9 & \textbf{181.4} & 183.5 \\
int8 & 64 & 8 & 400 & 422.0 & 442.1 & 436.8 & 449.9 & \textbf{338.3} & 341.1 & 345.5 & 350.0 & 340.8 & 343.8 \\
int8 & 128 & 8 & 50 & 109.9 & 115.3 & 108.6 & 114.9 & 104.5 & 107.4 & 104.7 & 109.1 & \textbf{95.9} & 98.1 \\
int8 & 128 & 8 & 100 & 183.2 & 188.4 & 264.8 & 280.6 & 178.5 & 180.8 & 175.6 & 180.3 & \textbf{166.4} & 168.2 \\
int8 & 128 & 8 & 200 & 407.0 & 418.9 & 406.5 & 420.9 & 313.8 & 316.4 & 314.4 & 319.3 & \textbf{307.1} & 310.4 \\
int8 & 128 & 8 & 400 & 680.5 & 692.2 & 685.6 & 700.7 & \textbf{586.5} & 590.3 & 595.6 & 600.5 & 588.0 & 590.6 \\
int8 & 192 & 8 & 50 & 245.0 & 259.4 & 245.2 & 260.4 & 154.2 & 157.6 & 155.2 & 160.1 & \textbf{146.4} & 149.3 \\
int8 & 192 & 8 & 100 & 357.6 & 371.5 & 364.7 & 377.0 & 266.0 & 270.0 & 271.5 & 276.4 & \textbf{262.8} & 265.7 \\
int8 & 192 & 8 & 200 & 591.0 & 603.0 & 595.8 & 610.1 & \textbf{498.4} & 502.2 & 505.0 & 509.8 & 500.2 & 504.2 \\
int8 & 192 & 8 & 400 & 1051.3 & 1067.6 & 1063.3 & 1077.1 & \textbf{958.1} & 967.9 & 971.1 & 976.3 & 959.8 & 964.4 \\
int8 & 128 & 1 & 100 & 96.9 & 102.2 & 94.1 & 100.2 & 92.6 & 95.8 & 90.2 & 94.6 & \textbf{82.9} & 86.8 \\
int8 & 128 & 32 & 100 & 601.7 & 615.4 & 596.2 & 626.4 & 510.5 & 513.4 & 507.0 & 529.3 & \textbf{500.4} & 525.9 \\
pm1 & 64 & 8 & 100 & 108.3 & 174.2 & 103.2 & 109.4 & 104.4 & 107.2 & 100.0 & 104.3 & \textbf{92.1} & 94.7 \\
pm1 & 128 & 8 & 100 & 217.1 & 234.3 & 220.9 & 236.8 & 129.4 & 132.5 & 130.6 & 135.0 & \textbf{124.6} & 130.1 \\
pm1 & 256 & 8 & 100 & 306.0 & 320.3 & 304.6 & 320.6 & 215.7 & 219.0 & 215.3 & 220.1 & \textbf{206.8} & 210.5 \\
pm1 & 512 & 8 & 100 & 541.1 & 555.8 & 527.8 & 550.2 & 450.9 & 453.9 & 437.0 & 451.5 & \textbf{431.5} & 447.4 \\
ternary & 256 & 8 & 100 & 345.6 & 358.6 & 346.5 & 363.1 & 252.7 & 256.2 & 256.3 & 263.3 & \textbf{249.6} & 255.5 \\
\bottomrule\end{tabular}

\end{table*}

\begin{table*}[tp]\centering\footnotesize
\caption{Decomposition of the unloaded round trip for the reference configuration (int8, $n=128$, $B=8$, $K=100$), in $\mu$s, from the servers' timestamps. ``Server I/O'' is the CPU time between \code{recv()} and \code{sendto()} minus the solve for the baselines, and the transmit wait for the GPU-initiated server. ``Residual'' is RTT minus the recorded server residence; its timing boundaries differ by server and it is not a direct wire-latency measurement.}
\label{tab:breakdown}
\begin{tabular}{lrrrrrrrrrr}\toprule
 & \multicolumn{2}{c}{queue} & \multicolumn{2}{c}{solve} & \multicolumn{2}{c}{server I/O} & \multicolumn{2}{c}{residual} & \multicolumn{2}{c}{RTT} \\
\cmidrule(lr){2-3}\cmidrule(lr){4-5}\cmidrule(lr){6-7}\cmidrule(lr){8-9}\cmidrule(lr){10-11}
server & p50 & p99 & p50 & p99 & p50 & p99 & p50 & p99 & p50 & p99 \\ \midrule
launch, default & 0.0 & 0.0 & 151.7 & 152.0 & 16.1 & 17.0 & 15.5 & 20.5 & 183.2 & 188.4 \\
proxy, default & 0.0 & 0.0 & 152.6 & 155.0 & 11.3 & 14.8 & 101.1 & 115.9 & 264.8 & 280.6 \\
launch, tuned & 0.0 & 0.0 & 151.7 & 152.0 & 16.1 & 17.0 & 10.7 & 12.6 & 178.5 & 180.8 \\
proxy, tuned & 0.0 & 0.0 & 152.9 & 155.4 & 12.0 & 15.5 & 10.5 & 12.5 & 175.6 & 180.3 \\
GPUNetIO & 2.0 & 2.1 & 149.9 & 151.1 & 0.8 & 2.0 & 13.5 & 15.7 & 166.4 & 168.2 \\
\bottomrule\end{tabular}

\end{table*}

\begin{figure}[!t]\centering
\includegraphics[width=\columnwidth]{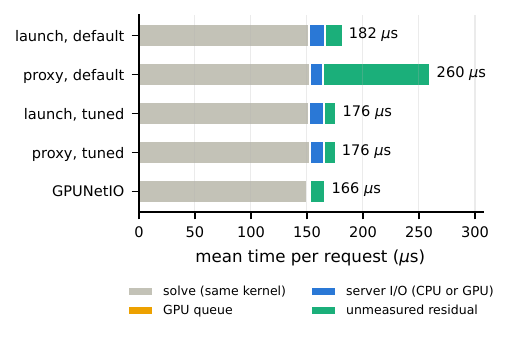}
\caption{Mean per-request timing components for the reference request. Means are used because they add to the mean RTT; component medians do not generally sum to the median RTT. The residual includes uninstrumented communication and endpoint work.}
\label{fig:breakdown}
\end{figure}

\begin{figure*}[tp]\centering
\includegraphics[width=\textwidth]{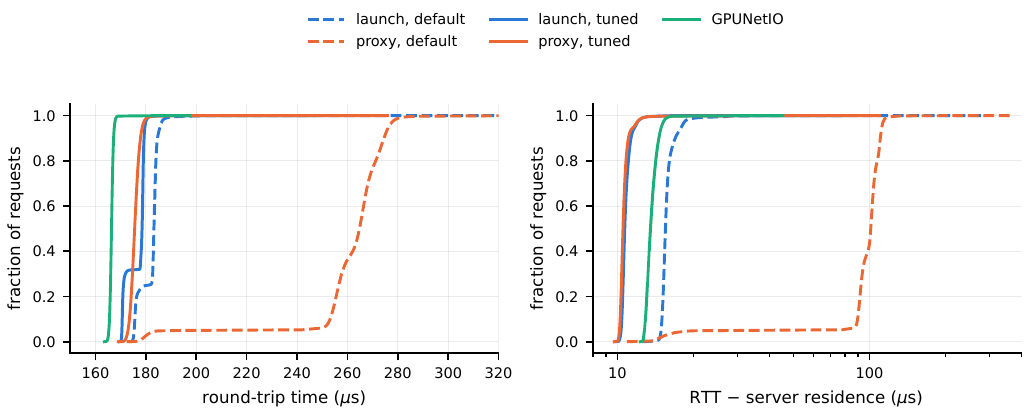}
\caption{Distribution of the unloaded round-trip time (left) and of the round trip minus the server's residence (right) for the reference configuration, 100,000 requests per server. On the default host the proxy's residual lies in its higher-latency mode for 95\% of the requests of this run; on the tuned host both CPU-mediated servers have a single mode near 10.5~$\mu$s. The GPU-initiated distribution is shown for comparison.}
\label{fig:cdf}
\end{figure*}

\subsubsection{Host-state dependence}
Table~\ref{tab:decomp} separates recorded residence from the residual RTT. The default-host CPU-mediated runs exhibit fast and slow residual modes, with the latter contributing roughly 100~$\mu$s. Tuning idle states, interrupt moderation, affinity, and offloads brings the residual close to 11~$\mu$s throughout the grid. The large default-host gap therefore reflects the host configuration, and we use the tuned servers for the comparison.

The GPU-initiated receive--solve--reply path bypasses the host network stack and its scheduling, and its unloaded RTT barely changed between the default and tuned host states (Section~\ref{sec:method}). The client is still CPU-based, so the load experiments below keep its receive path tuned.

\subsubsection{Comparison with tuned baselines}
The GPU-initiated service has the lowest recorded median in 15 of 19 settings, improving on the faster tuned baseline by 1.5--9.2~$\mu$s; the remaining four settings differ by 1.5--2.5~$\mu$s in favor of a baseline. Its recorded server residence is lower throughout the grid. For the reference request, the GPUNetIO residual is 13.5~$\mu$s, approximately 3~$\mu$s above the tuned baselines; the residuals use different timestamp boundaries (Section~\ref{sec:server}).

\begin{table*}[tp]\centering\footnotesize
\caption{Server residence (as stamped by the server) and the rest of the round trip (RTT minus residence), medians in $\mu$s, closed loop with one outstanding request.}
\label{tab:decomp}
\begin{tabular}{lrrrrrrrrrrrrr}\toprule
 & & & & \multicolumn{2}{c}{launch, default} & \multicolumn{2}{c}{proxy, default} & \multicolumn{2}{c}{launch, tuned} & \multicolumn{2}{c}{proxy, tuned} & \multicolumn{2}{c}{GPUNetIO} \\
\cmidrule(lr){5-6}\cmidrule(lr){7-8}\cmidrule(lr){9-10}\cmidrule(lr){11-12}\cmidrule(lr){13-14}
format & $n$ & $B$ & $K$ & server & rest & server & rest & server & rest & server & rest & server & rest \\ \midrule
int8 & 64 & 8 & 50 & 57.6 & 14.1 & 59.7 & 14.1 & 57.1 & 10.5 & 60.2 & 10.6 & 48.7 & 13.3 \\
int8 & 64 & 8 & 100 & 96.4 & 14.0 & 99.0 & 14.0 & 96.0 & 10.5 & 99.4 & 10.6 & 88.3 & 13.3 \\
int8 & 64 & 8 & 200 & 180.3 & 103.2 & 177.4 & 100.8 & 172.3 & 10.4 & 177.9 & 10.5 & 167.7 & 13.6 \\
int8 & 64 & 8 & 400 & 328.1 & 93.5 & 334.4 & 102.8 & 327.6 & 10.6 & 334.9 & 10.5 & 326.9 & 13.7 \\
int8 & 128 & 8 & 50 & 94.3 & 15.5 & 94.0 & 14.4 & 94.0 & 10.5 & 94.0 & 10.7 & 82.9 & 13.0 \\
int8 & 128 & 8 & 100 & 167.8 & 15.5 & 163.9 & 101.1 & 167.8 & 10.7 & 164.9 & 10.5 & 152.7 & 13.5 \\
int8 & 128 & 8 & 200 & 303.2 & 103.8 & 304.6 & 102.3 & 303.1 & 10.6 & 303.4 & 10.8 & 293.5 & 13.5 \\
int8 & 128 & 8 & 400 & 576.0 & 104.1 & 584.4 & 101.4 & 575.8 & 10.7 & 584.8 & 10.7 & 574.7 & 13.3 \\
int8 & 192 & 8 & 50 & 143.5 & 101.5 & 144.1 & 101.2 & 143.3 & 10.9 & 144.1 & 10.9 & 132.8 & 13.5 \\
int8 & 192 & 8 & 100 & 256.1 & 101.5 & 260.5 & 104.5 & 254.4 & 11.0 & 260.5 & 10.9 & 249.4 & 13.4 \\
int8 & 192 & 8 & 200 & 487.6 & 103.5 & 493.7 & 102.0 & 487.4 & 11.0 & 493.9 & 11.0 & 483.7 & 16.4 \\
int8 & 192 & 8 & 400 & 953.8 & 99.5 & 960.0 & 103.5 & 946.6 & 11.0 & 959.9 & 11.1 & 946.0 & 13.7 \\
int8 & 128 & 1 & 100 & 82.0 & 14.9 & 79.3 & 14.6 & 81.9 & 10.7 & 79.5 & 10.6 & 68.8 & 14.0 \\
int8 & 128 & 32 & 100 & 499.8 & 101.9 & 495.1 & 102.7 & 499.7 & 10.7 & 496.2 & 10.7 & 485.8 & 13.7 \\
pm1 & 64 & 8 & 100 & 94.3 & 14.0 & 89.1 & 14.0 & 93.9 & 10.5 & 89.4 & 10.5 & 78.7 & 13.3 \\
pm1 & 128 & 8 & 100 & 118.9 & 98.1 & 120.2 & 101.2 & 118.8 & 10.6 & 119.8 & 10.7 & 109.1 & 15.3 \\
pm1 & 256 & 8 & 100 & 204.9 & 101.1 & 204.3 & 100.8 & 204.7 & 10.9 & 204.2 & 11.0 & 193.0 & 13.7 \\
pm1 & 512 & 8 & 100 & 440.0 & 101.7 & 426.3 & 102.2 & 439.9 & 11.0 & 426.0 & 10.8 & 417.9 & 13.4 \\
ternary & 256 & 8 & 100 & 241.8 & 103.8 & 245.2 & 101.7 & 241.6 & 11.1 & 245.2 & 10.9 & 235.2 & 13.9 \\
\bottomrule\end{tabular}

\end{table*}

\subsubsection{Interpreting the timing decomposition}
Relative to GPUNetIO, the tuned launch and proxy servers have median residence differences of 9.6 and 10.7~$\mu$s across the 19 settings. These differences combine dispatch, copies or mapped-memory access, and the different timestamp boundaries, and they are close to the launch-path overhead measured without the network (Section~\ref{sec:gpucost}).

\subsection{Scaling with solve cost}
Fig.~\ref{fig:scaling} and Table~\ref{tab:ratio} follow median RTT as the solve grows, through $K$ at fixed $n$ and through $n$ at fixed $K$. With the tuned hosts, the CPU-mediated-to-GPUNetIO ratios remain near unity, spanning approximately 0.99--1.14. Dispatch savings matter most for short requests; as the common solver occupies more of the round trip, the ratios approach one. The larger separation under default host settings follows the slow receive mode discussed above. Without load, the main gains are therefore the freed server CPU core and lower host sensitivity; the larger latency difference appears under concurrent load.

\begin{figure*}[tp]\centering
\includegraphics[width=\textwidth]{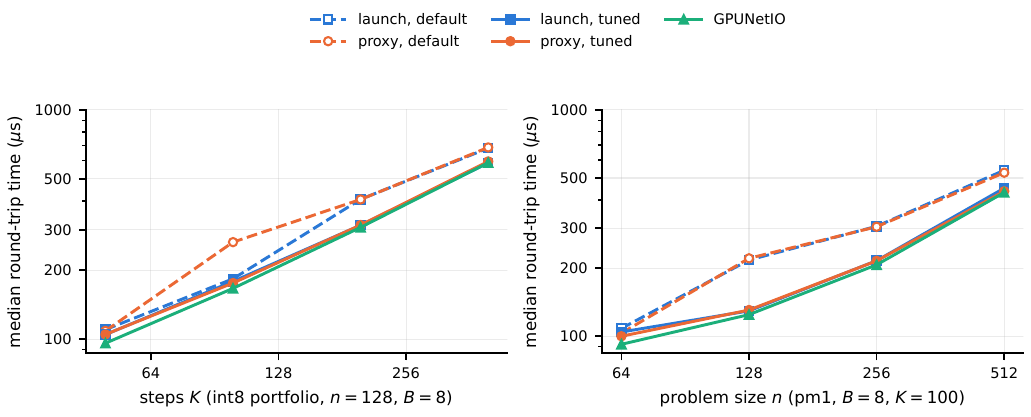}
\caption{Median unloaded round-trip time versus integration steps $K$ (left: \code{INT8}, $n=128$, $B=8$) and problem size $n$ (right: \code{PM1}, $B=8$, $K=100$). Dashed lines and open markers denote default-host CPU-mediated servers; solid lines and filled markers denote tuned CPU-mediated servers and GPUNetIO. Both axes use logarithmic scales.}
\label{fig:scaling}
\end{figure*}

\begin{table*}[tp]\centering\footnotesize
\caption{Median round-trip time of the CPU-mediated servers relative to the GPU-initiated server, on the default and on the tuned host, with the GPU-initiated server's residence as the measure of solve cost.}
\label{tab:ratio}
\begin{tabular}{lrrrrrrrr}\toprule
 & & & & GPUNetIO & \multicolumn{2}{c}{default host} & \multicolumn{2}{c}{tuned host} \\
\cmidrule(lr){6-7}\cmidrule(lr){8-9}
format & $n$ & $B$ & $K$ & server p50 ($\mu$s) & launch & proxy & launch & proxy \\ \midrule
int8 & 64 & 8 & 50 & 48.7 & 1.16 & 1.19 & 1.09 & 1.14 \\
int8 & 64 & 8 & 100 & 88.3 & 1.09 & 1.12 & 1.05 & 1.08 \\
int8 & 64 & 8 & 200 & 167.7 & 1.56 & 1.53 & 1.01 & 1.04 \\
int8 & 64 & 8 & 400 & 326.9 & 1.24 & 1.28 & 0.99 & 1.01 \\
int8 & 128 & 8 & 50 & 82.9 & 1.15 & 1.13 & 1.09 & 1.09 \\
int8 & 128 & 8 & 100 & 152.7 & 1.10 & 1.59 & 1.07 & 1.06 \\
int8 & 128 & 8 & 200 & 293.5 & 1.33 & 1.32 & 1.02 & 1.02 \\
int8 & 128 & 8 & 400 & 574.7 & 1.16 & 1.17 & 1.00 & 1.01 \\
int8 & 192 & 8 & 50 & 132.8 & 1.67 & 1.68 & 1.05 & 1.06 \\
int8 & 192 & 8 & 100 & 249.4 & 1.36 & 1.39 & 1.01 & 1.03 \\
int8 & 192 & 8 & 200 & 483.7 & 1.18 & 1.19 & 1.00 & 1.01 \\
int8 & 192 & 8 & 400 & 946.0 & 1.10 & 1.11 & 1.00 & 1.01 \\
int8 & 128 & 1 & 100 & 68.8 & 1.17 & 1.14 & 1.12 & 1.09 \\
int8 & 128 & 32 & 100 & 485.8 & 1.20 & 1.19 & 1.02 & 1.01 \\
pm1 & 64 & 8 & 100 & 78.7 & 1.18 & 1.12 & 1.13 & 1.09 \\
pm1 & 128 & 8 & 100 & 109.1 & 1.74 & 1.77 & 1.04 & 1.05 \\
pm1 & 256 & 8 & 100 & 193.0 & 1.48 & 1.47 & 1.04 & 1.04 \\
pm1 & 512 & 8 & 100 & 417.9 & 1.25 & 1.22 & 1.05 & 1.01 \\
ternary & 256 & 8 & 100 & 235.2 & 1.38 & 1.39 & 1.01 & 1.03 \\
\bottomrule\end{tabular}

\end{table*}

\subsection{Latency under load}\label{sec:load}
Tables~\ref{tab:window} and~\ref{tab:load} and Fig.~\ref{fig:load} report the reference configuration under closed-loop concurrency and open-loop arrivals. With the default client, the residual RTT increases under load across all servers while their residence times remain comparatively stable. We therefore use the tuned client below.

At 400--800~kreq/s offered, Table~\ref{tab:load} and Fig.~\ref{fig:load} report medians over repeated runs; Table~\ref{tab:repeats} lists the counts and ranges. At 800~kreq/s offered, GPUNetIO completes 618--628~kreq/s and the tuned proxy 620--668~kreq/s without loss; the client limits injection at this load. GPUNetIO's median RTT is then 189~$\mu$s against 288~$\mu$s for the tuned proxy, and its 99th-percentile RTT 218~$\mu$s against 609~$\mu$s, reductions by factors of 1.5 and 2.8. At 600~kreq/s the corresponding values are 188 against 240~$\mu$s and 216 against 327~$\mu$s; at 400~kreq/s the two are still close (178 against 193~$\mu$s and 207 against 229~$\mu$s). Between 400 and 800~kreq/s offered, GPUNetIO's median RTT rises by 6\% and the proxy's by 50\%. At the two highest rates, where the client is saturated, individual runs of both servers occasionally show millisecond-scale 99th percentiles; at 700~kreq/s this places the median 99th percentile for GPUNetIO (423.5~$\mu$s) slightly above that of the proxy (381~$\mu$s). The launch server reaches its four-stream limit much earlier, near 24~kreq/s.

The server timestamps help explain this behavior. From 400 to 800~kreq/s offered, GPUNetIO's median server residence stays at 156--159~$\mu$s and its 99th percentile at 180--186~$\mu$s, whereas the proxy's rise from 173 to 235~$\mu$s and from 208 to 317~$\mu$s. In closed loop, the same pipeline also supports a higher completion rate than the proxy at the largest tested window.

\begin{table*}[tp]\centering\footnotesize
\caption{Closed loop with 1--64 outstanding requests, reference configuration.}
\label{tab:window}
\begin{tabular}{lrrrrrrr}\toprule
server & outstanding & kreq/s & \multicolumn{2}{c}{server residence ($\mu$s)} & \multicolumn{3}{c}{client RTT ($\mu$s)} \\
\cmidrule(lr){4-5}\cmidrule(lr){6-8}
 & & & p50 & p99 & p50 & p99 & p99.9 \\ \midrule
launch, default & 1 & 5.5 & 167.8 & 168.8 & 183.2 & 188.4 & 197.6 \\
 & 4 & 21.4 & 161.3 & 169.8 & 178.6 & 278.7 & 324.1 \\
 & 16 & 24.0 & 162.6 & 166.5 & 667.0 & 765.0 & 772.6 \\
 & 64 & 24.1 & 162.4 & 166.4 & 2658.7 & 2760.1 & 2772.7 \\
\addlinespace
proxy, default & 1 & 3.8 & 163.9 & 168.2 & 264.8 & 280.6 & 300.2 \\
 & 4 & 20.4 & 163.4 & 172.0 & 178.9 & 293.5 & 322.3 \\
 & 16 & 66.6 & 167.4 & 198.2 & 237.1 & 349.5 & 440.6 \\
 & 64 & 160.0 & 175.4 & 223.6 & 416.1 & 548.1 & 619.3 \\
\addlinespace
launch, tuned & 1 & 5.7 & 167.8 & 168.7 & 178.5 & 180.8 & 184.4 \\
 & 4 & 22.7 & 162.0 & 170.2 & 175.9 & 184.8 & 190.2 \\
 & 16 & 24.4 & 162.0 & 165.6 & 654.3 & 661.9 & 665.7 \\
 & 64 & 24.4 & 162.0 & 165.7 & 2618.3 & 2629.5 & 2635.2 \\
\addlinespace
proxy, tuned & 1 & 5.7 & 164.9 & 169.3 & 175.6 & 180.3 & 186.5 \\
 & 4 & 22.9 & 163.2 & 168.2 & 174.0 & 180.0 & 194.0 \\
 & 16 & 91.3 & 162.5 & 182.8 & 174.2 & 194.6 & 204.3 \\
 & 64 & 339.1 & 168.9 & 198.1 & 185.7 & 215.7 & 223.9 \\
\addlinespace
GPUNetIO & 1 & 6.0 & 152.7 & 154.0 & 166.4 & 168.2 & 177.0 \\
 & 4 & 23.9 & 153.4 & 156.0 & 167.2 & 170.2 & 178.1 \\
 & 16 & 95.4 & 153.1 & 174.4 & 166.9 & 188.5 & 196.2 \\
 & 64 & 363.1 & 155.2 & 178.9 & 174.5 & 198.8 & 205.4 \\
\addlinespace
GPUNetIO, per-reply send & 1 & 5.9 & 152.9 & 154.3 & 168.6 & 172.4 & 178.2 \\
 & 4 & 23.6 & 153.3 & 154.7 & 169.2 & 173.7 & 187.3 \\
 & 16 & 94.2 & 152.6 & 171.7 & 169.1 & 188.6 & 199.0 \\
 & 64 & 326.5 & 165.8 & 195.7 & 191.4 & 221.5 & 268.6 \\
\addlinespace
\bottomrule\end{tabular}

\end{table*}

\begin{table*}[tp]\centering\footnotesize
\caption{Open loop with Poisson arrivals, reference configuration, 100,000 requests per run. ``Lost'' are requests still unanswered when the client completes its final receive drain, which waits up to 2~s after sending the last request. ``GPUNetIO, per-reply send'' is the same server with the per-reply send path of the vendor sample (and the 5-$\mu$s receive timeout of the first runs). For GPUNetIO and the tuned proxy at 400--800~kreq/s, each entry is the median over repeated runs (Table~\ref{tab:repeats}). At 700 and 800~kreq/s offered the client itself emitted 596--681~kreq/s.}
\label{tab:load}
\begin{tabular}{lrrrrrrrr}\toprule
server & offered & achieved & lost & \multicolumn{2}{c}{server residence ($\mu$s)} & \multicolumn{3}{c}{client RTT ($\mu$s)} \\
\cmidrule(lr){5-6}\cmidrule(lr){7-9}
 & kreq/s & kreq/s & \% & p50 & p99 & p50 & p99 & p99.9 \\ \midrule
launch, default & 10 & 10.1 & 0.0 & 161.1 & 174.4 & 249.8 & 425.8 & 549.6 \\
 & 50 & 24.1 & 51.8 & 162.0 & 166.2 & 11844.4 & 12397.1 & 12507.5 \\
 & 100 & 23.9 & 76.0 & 162.1 & 166.6 & 11984.0 & 12352.3 & 12433.8 \\
 & 200 & 23.6 & 88.0 & 162.2 & 166.5 & 11918.4 & 12185.4 & 12246.8 \\
 & 400 & 23.0 & 94.0 & 162.3 & 166.7 & 11887.1 & 12099.1 & 12131.7 \\
 & 800 & 22.5 & 96.1 & 162.8 & 167.4 & 11898.8 & 12124.2 & 12152.1 \\
\addlinespace
proxy, default & 10 & 10.1 & 0.0 & 163.2 & 188.9 & 180.4 & 387.8 & 501.7 \\
 & 50 & 50.2 & 0.0 & 175.1 & 220.3 & 592.5 & 1479.1 & 1631.9 \\
 & 100 & 100.5 & 0.0 & 178.0 & 224.8 & 587.2 & 1072.8 & 1181.6 \\
 & 200 & 200.9 & 0.0 & 185.4 & 237.6 & 474.5 & 742.8 & 795.0 \\
 & 400 & 401.5 & 0.0 & 208.5 & 290.8 & 481.4 & 636.9 & 691.2 \\
 & 800 & 499.2 & 27.5 & 334.8 & 494.1 & 959.1 & 1169.2 & 1205.0 \\
\addlinespace
launch, tuned & 10 & 10.1 & 0.0 & 160.9 & 170.0 & 172.3 & 264.1 & 315.7 \\
 & 50 & 24.3 & 51.3 & 161.6 & 165.0 & 11467.3 & 11589.9 & 11600.4 \\
 & 100 & 24.2 & 75.7 & 161.8 & 165.2 & 11504.1 & 11608.7 & 11618.8 \\
 & 200 & 23.9 & 87.8 & 161.7 & 165.7 & 11544.8 & 11617.6 & 11621.7 \\
 & 400 & 23.4 & 93.9 & 161.8 & 165.6 & 11585.3 & 11629.4 & 11635.8 \\
 & 800 & 22.8 & 96.1 & 161.8 & 165.7 & 11598.5 & 11646.0 & 11650.1 \\
\addlinespace
proxy, tuned & 10 & 10.1 & 0.0 & 162.9 & 176.7 & 173.6 & 188.1 & 201.4 \\
 & 50 & 50.3 & 0.0 & 161.6 & 185.9 & 172.9 & 197.3 & 204.5 \\
 & 100 & 100.5 & 0.0 & 162.2 & 188.3 & 174.2 & 200.5 & 207.2 \\
 & 200 & 201.0 & 0.0 & 164.5 & 193.3 & 178.4 & 208.1 & 215.2 \\
 & 400 & 401.8 & 0.0 & 173.4 & 207.6 & 192.6 & 229.0 & 242.3 \\
 & 600 & 602.4 & 0.0 & 207.0 & 264.4 & 239.8 & 327.0 & 363.0 \\
 & 700 & 627.9 & 0.0 & 222.1 & 287.4 & 263.2 & 381.2 & 614.0 \\
 & 800 & 641.6 & 0.0 & 235.2 & 316.8 & 288.3 & 608.5 & 699.1 \\
\addlinespace
GPUNetIO & 10 & 10.1 & 0.0 & 153.4 & 156.8 & 166.7 & 171.2 & 178.3 \\
 & 50 & 50.3 & 0.0 & 153.1 & 166.6 & 166.5 & 180.8 & 195.6 \\
 & 100 & 100.5 & 0.0 & 153.3 & 178.2 & 167.7 & 193.0 & 200.1 \\
 & 200 & 201.0 & 0.0 & 154.4 & 178.7 & 170.9 & 195.7 & 202.4 \\
 & 400 & 401.9 & 0.0 & 156.2 & 183.6 & 177.8 & 207.0 & 213.7 \\
 & 600 & 602.5 & 0.0 & 159.0 & 186.1 & 188.1 & 215.8 & 235.3 \\
 & 700 & 628.2 & 0.0 & 159.0 & 185.7 & 188.8 & 423.5 & 508.3 \\
 & 800 & 622.5 & 0.0 & 158.7 & 180.1 & 188.6 & 217.9 & 222.9 \\
\addlinespace
GPUNetIO, per-reply send & 10 & 10.1 & 0.0 & 153.4 & 160.6 & 169.1 & 176.9 & 195.3 \\
 & 50 & 50.3 & 0.0 & 152.9 & 177.7 & 169.5 & 194.6 & 203.7 \\
 & 100 & 100.5 & 0.0 & 153.4 & 180.5 & 171.4 & 199.0 & 206.7 \\
 & 200 & 201.0 & 0.0 & 157.3 & 186.0 & 178.2 & 209.1 & 218.8 \\
 & 400 & 401.5 & 0.0 & 249.9 & 484.9 & 311.5 & 537.0 & 564.6 \\
 & 800 & 333.7 & 40.9 & 23323.4 & 24793.2 & 24307.5 & 25902.0 & 25946.6 \\
\addlinespace
\bottomrule\end{tabular}

\end{table*}

\begin{table*}[tp]\centering\footnotesize
\caption{Repeated open-loop runs of the reference configuration with the current code (GPUNetIO and tuned proxy). Client RTT range over the runs, median in parentheses. One GPUNetIO run at 700~kreq/s, in which the client dropped 1.1\% of the replies although the server transmitted all of them, is excluded.}
\label{tab:repeats}
\begin{tabular}{lrrrrrr}\toprule
server & offered & runs & achieved & \multicolumn{3}{c}{client RTT ($\mu$s), range over runs (median)} \\
\cmidrule(lr){5-7}
 & kreq/s & & kreq/s & p50 & p99 & p99.9 \\ \midrule
GPUNetIO & 400 & 4 & 402 & 177--179 (178) & 204--209 (207) & 213--214 (214) \\
 & 600 & 3 & 602 & 187--188 (188) & 216--321 (216) & 220--531 (235) \\
 & 700 & 2 & 623--633 & 188--190 (189) & 206--641 (423) & 216--800 (508) \\
 & 800 & 4 & 618--628 & 187--189 (189) & 195--3127 (218) & 199--3326 (223) \\
\addlinespace
proxy, tuned & 400 & 4 & 402 & 191--194 (193) & 226--232 (229) & 236--247 (242) \\
 & 600 & 3 & 602 & 236--246 (240) & 313--347 (327) & 353--491 (363) \\
 & 700 & 3 & 617--646 & 262--286 (263) & 354--1427 (381) & 387--1598 (614) \\
 & 800 & 4 & 620--668 & 239--314 (288) & 408--876 (609) & 443--1032 (699) \\
\addlinespace
\bottomrule\end{tabular}

\end{table*}

\begin{figure*}[tp]\centering
\includegraphics[width=\textwidth]{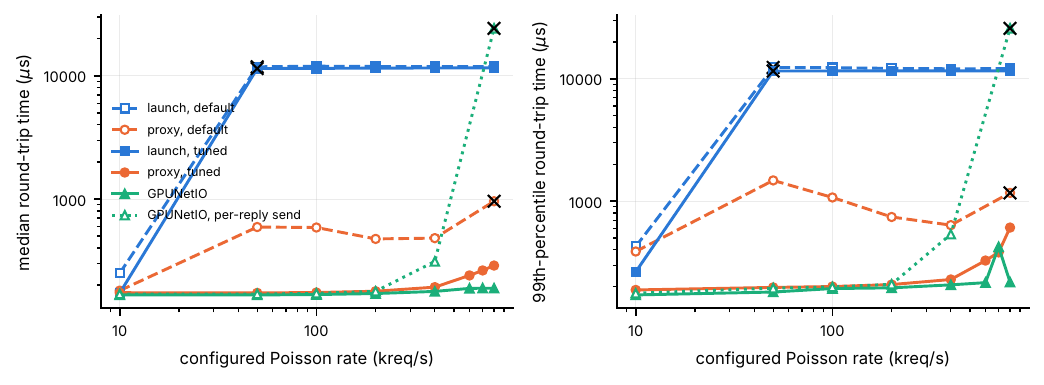}
\caption{Median (left) and 99th-percentile (right) round-trip time against the offered Poisson load, reference configuration; crosses mark the first run with lost requests. The four-stream launch server completes about 24~kreq/s under overload. At the highest offered load, GPUNetIO with per-reply sends completes 334~kreq/s with 41\% loss; batched GPUNetIO and the tuned proxy have no observed loss in the retained runs. Points of GPUNetIO and the tuned proxy at 400--800~kreq/s are medians over repeated runs. At the highest offered load, client injection falls below the configured rate.}
\label{fig:load}
\end{figure*}

Batched transmission keeps reply submission from limiting the request service. At the highest offered load, per-reply sends complete 334~kreq/s with a median RTT of 24~ms and 41\% request loss (Table~\ref{tab:load}). Batched descriptor posting avoids this bottleneck at the tested rates.

\section{Streaming dense couplings}\label{sec:stream}
\subsection{Workload and transfer volume}
The request service sends small field updates to a resident model. When the coupling matrix exceeds GPU memory, the streaming solver receives the matrix tile by tile. Positions, momenta, signs, and row accumulators remain on the receiving GPU, where workers consume arriving tiles to evaluate the dSB coupling force. A remote matrix generator provides a reproducible source at large scale. We measure coefficient reception together with coupling evaluation, including reuse of each tile across replicas.

A dense one-bit representation has logical size $n^2/8$ bytes. Frames contain a 64-byte header, including Ethernet, and $P_f=4032$ payload bytes, representing $C_b=32{,}256$ couplings from one row. The final column block is padded, so the transmitted payload is
\begin{equation}
V(n)=n\left\lceil n/C_b\right\rceil P_f.
\label{eq:volume}
\end{equation}
For $n=10^7$, the logical matrix is 12.5~TB, whereas a pass carries 12.53952~TB in $3.11\times10^9$ frames. For $n=10^4$, padding expands 12.5~MB to 40.32~MB. All streaming rates below count transmitted payload, including padding. For payload rate $R$ in bits/s, the transfer estimate is $t_{\mathrm{pass}}\approx8V/R$. A 4096-byte frame occupies 4120 bytes including FCS, preamble, and interpacket gap, giving an ideal payload ceiling of $400(4032/4120)=391.46$~Gb/s. On the current PCIe~5.0 server, the observed 307~Gb/s reaches approximately 78\% of this payload ceiling while performing the coupling accumulation; the topology and bandwidth scaling directions are discussed in Section~\ref{sec:bandwidth}.

\subsection{Generator and receiver}
The streaming runs use seed 42, $\Delta t=1.1$, and $\xi=0.5/\sqrt n$. The $n=10^7$, $B=1$ run uses two updates; its $B=8$ runs use ten. Reference comparisons at $n=10^4$ and $10^5$ use 20 updates. The generator runs on one GH200 GPU. Figure~\ref{fig:stream} summarizes the data and control paths. Synthesis blocks form frames in column-block-major order, and sender blocks post batches over one to eight queues. A pass ends with one \code{PASS\_END} per queue; the next pass starts after an acknowledgement. Random couplings are generated symmetrically from a hash of the seed and the unordered pair of indices. Mattis couplings are $J_{ij}=\eta_i\eta_j$ for $i\ne j$, where $\eta$ is a hashed sign vector. The transmitted diagonal is $+1$ and is removed during evaluation. Mattis instances have a known ground state; random instances exercise generic coupling accumulation.

The receiver has one cyclic queue containing 16,384 slots of 4096 bytes. A receive warp publishes bursts to 372 persistent worker blocks, each with 512 threads. Workers read frames in place, load the corresponding sign tile of each replica into shared memory, and accumulate XOR/popcount results into int32 row accumulators. At $n=10^7$ and $B=8$, the accumulators require 320~MB. The coupling sum is
\begin{equation}
a_i^{(b)}=n-2c_i^{(b)}-s_i^{(b)},
\end{equation}
where $c_i^{(b)}$ is the mismatch count over the full padded row. Padding bits in both operands are zero. The diagonal correction has the opposite sign from Proposition~\ref{prop:exact}(i), because that proposition stores a zero diagonal bit whereas the stream stores a one. For the measured range $n\le10^7$, the count fits int32 and $|a_i|\le n-1<2^{24}$, so conversion of the coupling sum to FP32 remains exact. The energy reduction uses 64-bit integers.

The host coordinates pass boundaries: it observes end markers, waits for published bursts to finish, and requests an in-kernel update through mapped control memory. Workers evaluate the energy of the \emph{incoming} signs, advance the dynamical state when an update is requested, and compute a checksum of the resulting sign bitmap. Thus an updating record contains pre-update energy and post-update checksum. A run of $K$ updates requires $K+1$ matrix passes, with the last pass evaluating the final energy without advancing the state. Packet reception and coupling accumulation execute on the GPU, with host coordination at pass boundaries.

The reported receive-pass interval ends when the host observes the end markers, which it polls every 2~ms. Full iteration time also includes draining outstanding work, updating the state, and acknowledging the pass; these control operations use host polling.

Work-ring acknowledgements protect work-item reuse. The cyclic packet buffer is managed separately, so workers must finish consuming each burst before the NIC wraps around. At the measured rate, the buffer holds approximately 1.7~ms of arriving frames. The evaluated implementation uses mapped host memory for its receive/control counters.

\begin{figure}[!t]\centering
\includegraphics[width=\columnwidth]{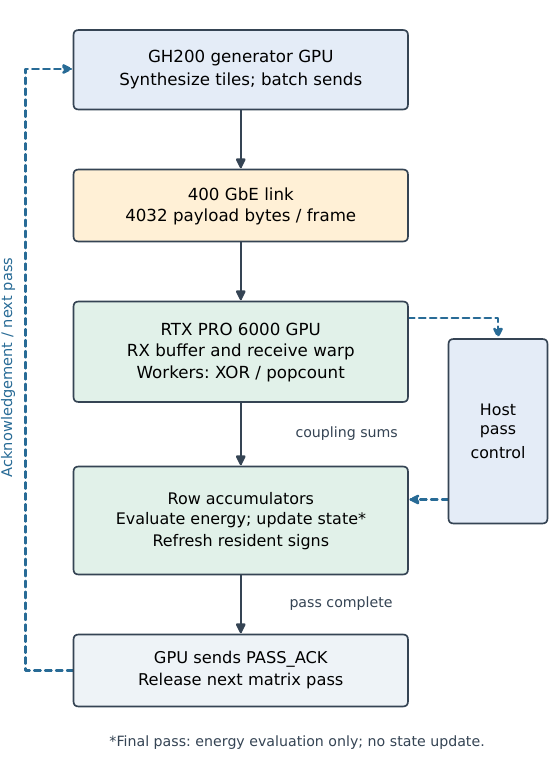}
\caption{Streaming implementation. Solid arrows show coefficient delivery and GPU computation. Dashed arrows show pass-boundary coordination and the GPU-sent acknowledgement that releases the generator for the next pass. Host control follows reception and draining of outstanding work. Each of the $K$ updates requires a matrix pass, followed by one final-energy pass.}
\label{fig:stream}
\end{figure}

\subsection{Validation}\label{sec:stream:correct}
Validation combines stream accounting, planted-solution checks, and reference execution. The receiver compares parsed stream-frame totals with the end markers and records dropped bursts. Mattis instances provide the exact ground-state value $2E=-n(n-1)$ for checking the accumulated energy and solver integration. At $n=10^5$, the Mattis and random energy/checksum records match a resident GPU reference byte for byte; at $n=10^4$, the random records also match a CPU reference. These records contain the best-replica energy and a 64-bit spin-bitmap checksum at each recorded step.

\subsection{Measured results}\label{sec:stream:meas}
Table~\ref{tab:stream} reports the streaming measurements. The $n=10^7$, $B=1$ Mattis run passes all recorded counter checks and reaches $2E=-99{,}999{,}990{,}000{,}000$ after its first update. Its later passes take 326.733~s at approximately 307.03~Gb/s. The 12.5-TB matrix is never stored on the receiver.

At $n=10^6$, the reported runs with $B=1$ and $B=8$ and one to eight sender queues have similar rates near 307~Gb/s. This supports amortizing matrix transmission across the tested replicas. The runs at $n=10^5$ and $10^4$ serve as reference checks. At $n=10^4$, a pass carries 40~MB, equivalent to about 1~ms at 307~Gb/s. The host polls for end markers every 2~ms, so the recorded 2.1-ms interval includes a substantial polling delay.

At $n=10^7$ with $B=8$, the Mattis run achieves 306.94--307.01~Gb/s and all eight replicas reach the known ground-state energy. The random instance achieves 306.20--306.41~Gb/s; its best final energy is $2E=-41{,}168{,}653{,}412$ after ten updates. Each update requires a 327-s pass. The random-instance run tests coupling evaluation over ten updates; it does not establish convergence. Every pass contains the expected $3.11\times10^9$ parsed stream frames, with zero reported missing stream frames and zero dropped bursts.

\begin{table*}[tp]\centering\footnotesize
\caption{Streaming observations from the preserved records. Rate and receive-pass time are ranges after excluding pass 0, whose interval includes the start-up of the generator, and, for $n\le10^5$, pass 1, which contains a one-time delay of about 10~ms. Payload includes padding. All listed passes have matching parsed stream counts and zero dropped bursts. Reference matches concern energy/checksum records.}
\label{tab:stream}
\begin{tabular}{lrrrrrl}\toprule
Instance & $n$ & $B$ & $N_q$ & Payload (Gb/s) & Receive pass (s) & Evidence \\ \midrule
Mattis & $10^7$ & 1 & 4 & 307.0 & 326.73 & $2E=-n(n-1)$ at step 1 \\
Mattis & $10^7$ & 8 & 4 & 306.9--307.0 & 326.75--326.83 & $2E=-n(n-1)$ at step 1 \\
random & $10^7$ & 8 & 4 & 306.2--306.4 & 327.39--327.61 & $2E=-4.12\times10^{10}$ at step 10 \\
Mattis & $10^6$ & 8 & 1--8 & 307.0--307.1 & 3.361--3.362 & $2E=-n(n-1)$ \\
Mattis & $10^6$ & 1 & 4 & 307.0 & 3.362 & $2E=-n(n-1)$ \\
random & $10^6$ & 8 & 4 & 307.0--307.1 & 3.361--3.362 & energy reported \\
Mattis & $10^5$ & 8 & 4 & 308.8--309.4 & 0.0417--0.0418 & GPU records match \\
random & $10^5$ & 8 & 4 & 308.8--311.5 & 0.0414--0.0418 & GPU records match \\
random & $10^4$ & 8 & 4 & 153.6--153.7 & 0.0021 & CPU records match \\
\bottomrule
\end{tabular}

\end{table*}

At $B=8$, each received coefficient contributes to eight coupling evaluations. The similar payload rates for one and eight replicas show that this reuse amortizes matrix transfer across the tested replica batch.

\section{Discussion}\label{sec:discussion}
\subsection{Request-rate and model-size scaling}
The resident-model and streaming modes place different demands on the same communication path. Small field updates require low scheduling overhead across many concurrent solves. Dense coupling tiles require sustained bandwidth while the solver state stays resident. In both cases, GPU workers consume coefficients directly from the receive path.

For resident models, the main gain appears under load. The unloaded latency is close to the tuned proxy's, but GPUNetIO's median RTT rises by 6\% from 400 to 800~kreq/s offered, compared with 50\% for the proxy. Each CPU-mediated baseline dedicates a polling core to receiving requests and returning results; GPUNetIO performs that work on the device after setup. Keeping packet handling on the GPU also reduces dependence on host scheduling and interrupt settings.

The request-per-block mapping allows independent solves to run concurrently, and batched descriptor posting drains their completed replies. Both are needed to keep communication and solver execution at comparable rates.

For model-size scaling, streaming separates quadratic coefficient storage from the resident solver state. A dense one-bit matrix requires $O(n^2)$ storage, whereas the dynamical state, sign bitmaps, and accumulators require $O(Bn)$. At $n=10^7$, the 12.5-TB logical matrix is consumed through a 64-MiB packet buffer, with 320~MB of row accumulators for eight replicas. Device memory holds the solver state and receive buffers; the time per update is then governed by coefficient-delivery bandwidth. Each incoming tile is reused across replicas, but every update requires another matrix pass.

\subsection{External coefficients and changing models}
Streaming separates coefficient production from solver execution. Potential sources include remote memory, storage arrays, and simulation or learning processes that construct successive models. An online application can update $h_{\ell}$ while reusing $J_{\ell}$, replace a resident model between requests, or supply a new sequence of dense tiles when the full model cannot be retained. When the coefficients admit a low-rank or sparse representation, preserving that structure offers an additional reduction in model traffic.

The measured 306--307~Gb/s corresponds to approximately 38.3~GB/s of coefficient payload. For scale, a Samsung 9100 PRO PCIe~5.0 NVMe SSD specifies sequential reads up to 14.8~GB/s \cite{ssd9100}; feeding the measured stream from such storage would call for parallel drives. The coefficient source must sustain this bandwidth to feed the receiver continuously.

\subsection{Topology and bandwidth scaling}\label{sec:bandwidth}
The current PCIe~5.0 server sustains 307~Gb/s. Its reported \code{SYS}/\code{NODE} GPU-to-NIC paths traverse host interconnects, making GPU/NIC placement one target for improving this rate. A direct PCIe-switch path, sufficient upstream bandwidth, and balanced queue placement are candidates for reducing the gap to the 400-GbE link's 391.46-Gb/s payload ceiling. The similar rates for one and eight replicas suggest room to increase delivery bandwidth while retaining replica reuse.

For payload volume $V$ and rate $R$, streaming incurs a transfer term $8V/R$ per pass. A trajectory of $K$ updates uses $K+1$ passes, including final-energy evaluation. The present $n=10^7$ configuration takes approximately 327~s per receive pass. At the ideal payload ceiling of the same 400-GbE link, the transfer term is approximately 256~s; an 800-GbE link with the same frame efficiency reduces this ideal term to approximately 128~s. These estimates describe bandwidth scaling at fixed model size. A distributed extension could partition matrix rows across $G$ receivers, each sustaining payload rate $R$. With sufficient aggregate source bandwidth, the ideal transfer term becomes $8V/(GR)$, plus synchronization and sign exchange; each replica's sign bitmap contains $n/8$ bytes.

PCIe~6.0 doubles the nominal bandwidth of PCIe~5.0 at the same lane count, increasing the per-lane data rate from 32 to 64~GT/s \cite{pcie6}. Combining this wider host interconnect with 800-Gb/s networking expands the available coefficient-delivery bandwidth. Newer ConnectX platforms provide this combination, and the BlueField-4 data processing unit (DPU) integrates PCIe~6.0, an 800-Gb/s network interface, programmable processing, and infrastructure acceleration \cite{supernic,bluefield}. The coefficient source, PCIe fabric, and GPU consumer must sustain the higher network rate. A DPU could handle communication control alongside GPU coupling evaluation.

\subsection{Scaling beyond per-packet CPU control}
The slowdown in single-thread processor performance growth strengthens the case for reducing frequent host coordination as network and PCIe bandwidth increase \cite{hennessy2019}. This concern also applies after kernel bypass: in NVIDIA's Aerial example, CPU coordination became a bottleneck as the number of radio units increased in a DPDK GPUdev pipeline, motivating GPU-controlled reception with GPUNetIO \cite{gpunetioblog}. Queue service, completion handling, and CPU--GPU notification all contribute to this coordination cost. At the 4120-byte wire occupancy used here, 400 and 800~Gb/s correspond to approximately 12.1 and 24.3 million packets/s, respectively. A serial packet-processing stage would therefore have only 82.4 and 41.2~ns per packet to keep pace. Batching, multiple cores, and offloads help host-mediated designs meet this demand. GPU initiation moves frequent queue service and packet consumption to device-side agents, where each received tile is reused across parallel solver replicas. This reduces CPU--GPU notifications, while throughput still depends on topology, batching, queue placement, and GPU occupancy.

\subsection{Implementation scope}\label{sec:scope}
The request evaluation uses one platform and two socket-based baselines. A matched DPDK/GPUdev baseline using the same solver, packet format, and batching policy is a useful next comparison to separate GPU-initiation gains from kernel-bypass gains. Deployment extensions include synchronized replacement of in-use models, retransmission and per-frame sequence validation, and distributing streaming work across multiple receivers. The current request bounds remain one frame per solve, one 512-KiB model slot, and the register/shared-memory limits of Section~\ref{sec:solver}.

\section{Conclusion}\label{sec:conclusion}
We implemented dSB as a GPU-initiated request service and a solver for streamed dense couplings. The request service keeps receive, solve, and reply operations on the GPU, using one worker block per request and batched transmission. It requires no dedicated server CPU data-path core. From 400 to 800~kreq/s offered, median RTT rises by 6\%; at the highest tested load, median and 99th-percentile latencies are lower than those of the tuned persistent CPU proxy by factors of 1.5 and 2.8 across repeated runs.

The streaming solver retains $O(Bn)$ state and consumes the $O(n^2)$ coefficients tile by tile. It evaluates ten-million-variable binary coupling matrices at 306--307~Gb/s for one or eight replicas. Known Mattis ground states and smaller-scale reference comparisons validate the execution path. The two modes extend dSB along different axes: concurrent decisions with a resident model, and larger models supplied as coefficient streams. Faster interconnects and multiple receivers could reduce the matrix-pass time, provided the coefficient source and GPU consumers sustain the aggregate bandwidth.

\section*{Code and data availability}
The implementation and experimental data are available in the author's GitHub repository at \url{https://github.com/yaochengchen/dsb-gpunetio}. The companion GPU solver is available at \url{https://github.com/yaochengchen/dSB_GPUs_3Million_N} \cite{ourrepo}.

\section*{Acknowledgment}
The author acknowledges support from FAPESP (grants 25/25586-5 and 2021/02120-0). Generative AI tools (OpenAI ChatGPT, Anthropic Claude and Google Gemini) were used to assist with software development, and ChatGPT for language editing; the author reviewed and verified all AI-assisted output and is responsible for the manuscript and the software.

\section*{Conflict of Interest}
The author declares no conflicts of interest.

\raggedbottom

\begin{IEEEbiographynophoto}{Yaocheng Chen}
received the B.S. degree in physics and electronics from Central South University in 2015, and the M.S. and Ph.D. degrees in physics from National Taiwan University in 2017 and 2025, respectively. He is a FAPESP postdoctoral fellow at IAG-USP, Universidade de S\~ao Paulo, Brazil. Previously, he was an algorithm scientist at Jingdong, China; a project researcher and visiting scholar at the Quantum Science Center of Guangdong-Hong Kong-Macao Greater Bay Area; and a KASI postdoctoral fellow at the Korea Astronomy and Space Science Institute (KASI), Republic of Korea. His interests include high-energy astroparticle physics, extragalactic astrophysics, GPU computing and FPGA systems. He is a member of the ASTRI Mini-Array and TAROGE collaborations, the CTAO Consortium and the Space Array team. His previous collaborations include ANITA/PUEO, ARA and BURSTT; he was also a member of the ALMA Total Power GPU Spectrometer (TPGS) team. ORCID: \url{https://orcid.org/0000-0002-8967-4911}.
\end{IEEEbiographynophoto}

\EOD
\end{document}